\documentclass[11pt,pdftex,a4paper]{article}%
\pdfoutput=1
\usepackage{verbatim}
\usepackage{amssymb}
\usepackage{lipsum}
\usepackage{multicol}
\usepackage{amsmath}
\usepackage{framed}
\usepackage{float}
\floatstyle{boxed} 
\restylefloat{figure}
\usepackage{subfig}
\usepackage{amsthm}
\usepackage{tikz}
\usepackage{hyperref}
\usepackage{color}
\usepackage{enumitem}
\setlist{nolistsep}
\usepackage{tikz}
\usepackage[margin=1in]{geometry}
\usepackage{cite}
\usepackage[textsize=tiny]{todonotes}

\newif\ifcode
\codefalse
\codetrue
\usepackage{macros}
\ifcode
\usepackage[ruled]{algorithm}
\usepackage{algpseudocode}
\usepackage{ioa_code}

\newtheorem{theorem}{Theorem}

\newtheorem{claim}[theorem]{Claim}

\newtheorem{corollary}[theorem]{Corollary}
\newtheorem{definition}{Definition}
\newtheorem{lemma}[theorem]{Lemma}
\newtheorem{observation}[theorem]{Observation}

\newcounter{linenumber}

\def\M{\ensuremath{\mathcal{M}}}

\def\X{\ensuremath{\mathcal{X}}}

\def\Nat{\ensuremath{\mathbb{N}}}

\def\TS{\mathit{TS}}
\def\shared{\mathit{shared}}
\def\exclusive{\mathit{exclusive}}

\newcommand{\RNum}[1]{\uppercase\expandafter{\romannumeral #1\relax}}

\newcommand{\true}{\mathit{true}}
\newcommand{\false}{\mathit{false}}

\newcommand{\remove}[1]{}

\newcommand{\Wset}{\textit{Wset}}
\newcommand{\Rset}{\textit{Rset}}
\newcommand{\Dset}{\textit{Dset}}

\newcommand{\txns}{\textit{txns}}

\newcommand{\Read}{\textit{read}}
\newcommand{\Write}{\textit{write}}
\newcommand{\TryC}{\textit{tryC}}
\newcommand{\TryA}{\textit{tryA}}
\newcommand{\ok}{\textit{ok}}

\newcommand{\ignore}[1]{}

\begin{document}
\bibliographystyle{abbrv}

\title{Inherent Limitations of Hybrid Transactional Memory}
\author{
Dan Alistarh$^1$~~~Justin Kopinsky$^4$~~~Petr Kuznetsov$^{2}$~~~Srivatsan Ravi$^3$~~~Nir Shavit$^4$$^,$$^5$ \\
$^1$\normalsize Microsoft Research, Cambridge\\
$^2$\normalsize T\'el\'ecom ParisTech\\
$^3$\normalsize TU Berlin \\
$^4$\normalsize Massachusetts Institute of Technology \\
$^5$\normalsize Tel Aviv University
}

\date{}
\maketitle
\thispagestyle{empty}
\begin{abstract}
Several Hybrid Transactional Memory (HyTM) schemes have recently been proposed to complement the fast,
but best-effort nature of Hardware Transactional Memory (HTM) with a slow, reliable software backup.
However, 
the costs of providing concurrency between hardware and software
transactions in HyTM are still not well understood.

In this paper, we propose a general model for HyTM implementations, 
which captures the ability of hardware transactions to buffer memory
accesses.
The model allows us to formally quantify and analyze
the amount of overhead (instrumentation) 
caused by the potential presence of software transactions.
We prove that (1) it is impossible to build a strictly serializable HyTM implementation that has both
uninstrumented reads and writes, even for very weak progress guarantees, and 
(2) the instrumentation cost incurred by a hardware transaction in any progressive
opaque HyTM may get linear in the transaction's data set. 
We further describe two implementations that,  for two different
progress conditions, exhibit 
optimal instrumentation costs. 
In sum, this paper captures for the first time an inherent
trade-off between the degree of hardware-software TM concurrency 
and the amount of
incurred instrumentation overhead. 
\end{abstract}

\newpage
\pagenumbering{arabic}\setcounter{page}{1}
\renewcommand{\paragraph}[1]{ \vspace{0.1cm} \noindent\textbf{#1}\hspace{0.1em}}
\linespread{0.96}
%
\section{Introduction}
\label{sec:intro}
\paragraph{Hybrid transactional memory.}
Ever since its introduction by Herlihy and Moss~\cite{HM93}, \emph{Transactional Memory (TM)} has promised to be an extremely useful tool,  
with the power to fundamentally change concurrent programming. 
It is therefore not surprising that the recently introduced Hardware Transactional Memory (HTM) implementations~\cite{Rei12, asf, bluegene} have been eagerly anticipated and scrutinized by the community. 

Early experience with programming HTM, e.g.~\cite{DiceLMN09, DragojevicMLM11, AlistarhEMMS14}, paints an interesting picture: 
if used carefully, HTM can be an extremely useful construct, 
and can significantly speed up and simplify concurrent implementations. 
At the same time, this powerful tool is not without its limitations: 
since HTMs are usually implemented on top of the cache coherence mechanism, 
hardware transactions have inherent \emph{capacity constraints} on the number of distinct memory locations 
that can be accessed inside a single transaction.  
Moreover, all current proposals are \emph{best-effort}, as they may abort under imprecisely specified conditions. 
In brief, the programmer should not solely rely on HTMs.

Several \emph{Hybrid Transactional Memory (HyTM)} schemes~\cite{hybridnorec,damronhytm, kumarhytm,phasedtm} have been
proposed to complement the fast, but best-effort nature of HTM 
with a slow, reliable software transactional memory (STM) backup. 
These proposals have explored a wide range of trade-offs between the
overhead on hardware transactions, concurrent execution of hardware and
software, and the provided progress guarantees. 

Early proposals for HyTM implementations~\cite{damronhytm, kumarhytm} 
shared some interesting features.
First, transactions that do not conflict are expected to run
concurrently, regardless of their types (software or hardware).
This property is referred to as \emph{progressiveness}~\cite{tm-theory}
and is believed to allow for increased parallelism.
Second, in addition to
exchanging the values of transactional objects, hardware transactions usually employ \emph{code instrumentation} techniques.
Intuitively, instrumentation is used by hardware transactions to
detect concurrency scenarios and abort in the
case of contention.
The number of instrumentation steps performed by these implementations within a hardware
transaction is usually proportional to the size of the transaction's data set. 

Recent work by Riegel \emph{et al.}~\cite{hynorecriegel} surveyed the various HyTM algorithms to date, focusing on techniques to reduce instrumentation overheads in the frequently executed hardware fast-path. 
However, it is not clear whether there are fundamental limitations when building a HyTM with non-trivial concurrency between
hardware and software transactions. In particular, what are the inherent instrumentation costs of building a HyTM, and what are the trade-offs between these
costs and the provided \emph{concurrency}, \emph{i.e.}, the ability of the HyTM
system to run software and hardware transactions in parallel?

\paragraph{Modelling HyTM.}
To address these questions, 
we propose the first model for hybrid TM systems which formally captures the notion of
\emph{cached} accesses provided by hardware transactions, and
precisely defines instrumentation costs in a quantifiable way.

We model a hardware transaction as a series of
memory accesses that operate on locally cached copies of the variables, followed by a \emph{cache-commit} operation.
In case a concurrent transaction performs a (read-write or write-write) conflicting access
to a cached object, the cached copy is invalidated and the hardware transaction aborts.   

Our model for instrumentation is motivated by recent experimental evidence 
which suggests that the overhead on hardware transactions imposed by code which 
detects concurrent software transactions is a significant performance bottleneck~\cite{MS13}. 
In particular, we say that a HyTM implementation imposes a logical partitioning of shared memory 
into \emph{data} and \emph{metadata} locations. 
Intuitively, metadata is used by transactions to exchange information about contention and 
conflicts while data locations only store the \emph{values} of data items read and updated within transactions. 
We quantify instrumentation cost by measuring the number of accesses to \emph{metadata objects} which transactions perform. 

\paragraph{The cost of instrumentation.}
Once this general model is in place, we derive two lower bounds on the cost of implementing a HyTM.  
First, we show that some instrumentation is necessary in a HyTM implementation even if we only
intend to provide \emph{sequential} progress,
where a transaction is only guaranteed to commit if it runs in the
absence of concurrency. 

Second, we prove that any progressive HyTM implementation providing \emph{obstruction-free liveness} (every operation
running \emph{solo} returns some response) and  has
executions in which an arbitrarily long read-only hardware transaction running in the
absence of concurrency \emph{must} access a number of 
distinct metadata objects proportional to the size of its data set.
We match this lower bound with an HyTM
algorithm that, additionally, allows for uninstrumented writes and
\emph{invisible reads}.

\paragraph{Low-instrumentation HyTM.}
The high instrumentation costs of early HyTM designs, 
which we show to be inherent, stimulated more recent HyTM schemes~\cite{phasedtm,hybridnorec,hynorecriegel,MS13}  
to sacrifice progressiveness for \emph{constant} instrumentation cost (\emph{i.e.}, not
depending on the size of the transaction). 
In the past two years, Dalessandro \emph{et al.}~\cite{hybridnorec} and 
Riegel \emph{et al.}~\cite{hynorecriegel} have proposed HyTMs based on the efficient \emph{NOrec STM}~\cite{norec}. 
These HyTMs schemes do not guarantee any parallelism among transactions; only
sequential progress is ensured. 
Despite this, they are among the best-performing HyTMs to date due to
the limited instrumentation in hardware transactions. 

Starting from this observation, we provide a more precise upper bound
for \emph{low-instrumentation} HyTMs
by presenting a HyTM algorithm with invisible reads \emph{and} uninstrumented
hardware writes which guarantees that a hardware transaction accesses at most one metadata object in the course of its execution.
Software transactions in this implementation remain progressive, while hardware transactions are guaranteed to commit
only if they do not run concurrently with an updating software transaction (or exceed capacity).
Therefore, the cost of avoiding the linear lower bound for progressive implementations is that hardware
transactions may be aborted by non-conflicting software ones.

In sum, this paper captures for the first time an inherent
trade-off between the degree of concurrency between hardware and
software transactions provided a HyTM implementation
and the incurred amount of instrumentation overhead. 

\paragraph{Roadmap.} The rest of the paper is organized as follows. Section~\ref{sec:prel} introduces the basic TM model and definitions.
Section~\ref{sec:hytm} presents our model of HyTM implementations, and
Section~\ref{sec:ins}~formally defines instrumentation.
Sections~\ref{sec:main} proves the impossibility of implementing uninstrumented HyTMs, while
Section~\ref{sec:main2} establishes a linear tight bound on metadata accesses for progressive HyTMs.
Section~\ref{sec:main3} describes an algorithm that overcomes this linear cost by weakening progress.
Section~\ref{sec:rel} presents the related work and Section~\ref{sec:disc} concludes the paper.
The Appendix contains the pseudo-code of the algorithms presented in this paper and their proofs of correctness.
%
%


\section{Preliminaries}
\label{sec:prel}
\vspace{1mm}\noindent\textbf{Transactional Memory (TM).} 
A \emph{transaction} is a sequence of \emph{transactional operations}
(or \emph{t-operations}), reads and writes, performed on a set of \emph{transactional objects} 
(\emph{t-objects}). 
A transactional memory \emph{implementation} provides a set of
concurrent \emph{processes} with deterministic algorithms that implement reads and
writes on t-objects using  a set of \emph{base objects}.
%

More precisely, for each transaction $T_k$, a TM implementation must support the following t-operations: 
$\mathit{read}_k(X)$, where $X$ is a t-object, that returns a value in
a domain $V$
or a special value $A_k\notin V$ (\emph{abort}),
$\mathit{write}_k(X,v)$, for a value $v \in V$,
that returns $\mathit{ok}$ or $A_k$, and
$\mathit{tryC}_k$ that returns $C_k\notin V$ (\emph{commit}) or $A_k$.

\vspace{1mm}\noindent\textbf{Configurations and executions.} 
A \emph{configuration} of a TM implementation specifies the state of each base object and each process. 
In the \emph{initial} configuration, each base object has its initial value and each process is in its initial state. 
An \emph{event} (or \emph{step}) of a transaction invoked by some process is an invocation of a t-operation, 
a response of a t-operation, or an atomic \emph{primitive} operation applied to base object along with its response. 
An \emph{execution fragment} is a (finite or infinite) sequence of events $E = e_1,e_2,\dots$. 
An \emph{execution} of a TM implementation $\mathcal{M}$ is an
execution fragment where, informally, each event respects the
specification of base objects and the algorithms specified by $\mathcal{M}$.
In the next section, we define precisely how base objects should
behave in a hybrid model combining direct memory accesses with \emph{cached} accesses (hardware
transactions).

The \emph{read set} (resp., the \emph{write set}) of a transaction $T_k$ in an execution $E$,
denoted $\Rset_E(T_k)$ (and resp. $\Wset_E(T_k)$), is the set of t-objects that $T_k$ attempts to read (and resp. write) 
by issuing a t-read (and resp. t-write) invocation in $E$ (for brevity, we sometimes 
omit the subscript $E$ from the notation).
The \emph{data set} of $T_k$ is $\Dset(T_k)=\Rset(T_k)\cup\Wset(T_k)$.
$T_k$ is called \emph{read-only} if $\Wset(T_k)=\emptyset$; \emph{write-only} if $\Rset(T_k)=\emptyset$ and
\emph{updating} if $\Wset(T_k)\neq\emptyset$.
Note that we consider the conventional dynamic TM model: 
the data set of a transaction is not known apriori (\emph{i.e.}, at the start of the transaction)
and it is identifiable only by the set of t-objects the transaction has invoked a read or write in the given execution.

For any finite execution $E$ and execution fragment $E'$, $E\cdot E'$ denotes the concatenation of $E$ and $E'$
and we say that $E\cdot E'$ is an \emph{extension}
of $E$.
For every transaction identifier $k$,
$E|k$ denotes the subsequence of $E$ restricted to events of
transaction $T_k$.
If $E|k$ is non-empty,
we say that $T_k$ \emph{participates} in $E$,
and let $\txns(E)$ denote the set of transactions that participate in $E$.
Two executions $E$ and $E'$
are \emph{indistinguishable} to a set $\mathcal{T}$ of transactions, if
for each transaction $T_k \in \mathcal{T}$, $E|k=E'|k$.

\vspace{1mm}\noindent\textbf{Complete and incomplete transactions.}
A transaction $T_k\in \txns(E)$ is \emph{complete in $E$} if
$E|k$
ends with a response event.
The execution $E$ is \emph{complete} if all transactions in $\txns(E)$
are complete in $E$.
A transaction $T_k\in \txns(E)$ is \emph{t-complete} if $E|k$
ends with $A_k$ or $C_k$; otherwise, $T_k$ is \emph{t-incomplete}.
$T_k$ is \emph{committed} (resp.\ \emph{aborted}) in $E$
if the last event of $T_k$ is $C_k$ (resp.\ $A_k$).
The execution $E$ is \emph{t-complete} if all transactions in
$\txns(E)$ are t-complete.
A configuration $C$ after an execution $E$ is \emph{quiescent} (resp. \emph{t-quiescent}) if 
every transaction $T_k \in \ms{txns}(E)$ is complete (resp. t-complete) in $E$.

\vspace{1mm}\noindent\textbf{Contention.}
We assume that base objects are accessed with \emph{read-modify-write} (rmw) primitives~\cite{G05,Her91}. 
A rmw primitive $\langle g,h \rangle$ applied to a base object 
atomically updates the value of the object with a new value, which is
a function $g(v)$ of the old value $v$, and returns a response $h(v)$.
A rmw primitive event on a base object is \emph{trivial} if, in any configuration, its application
does not change the state of the object. 
Otherwise, it is called \emph{nontrivial}.

Events $e$ and $e'$ of an execution $E$  \emph{contend} on a base
object $b$ if they are both primitives on $b$ in $E$ and at least 
one of them is nontrivial.

In a configuration $C$ after an execution $E$, every incomplete transaction $T$  
has exactly one \emph{enabled} event in $C$, 
which is the next event $T$ will perform according to the TM implementation.

We say that a transaction $T$ is \emph{poised to apply an event $e$ after $E$} 
if $e$ is the next enabled event for $T$ in $E$.
We say that transactions $T$ and $T'$ \emph{concurrently contend on $b$ in $E$} 
if they are each poised to apply contending events on $b$ after $E$.

We say that an execution fragment $E$ is \emph{step contention-free} for t-operation $op_k$ if the events of $E|op_k$ 
are contiguous in $E$.
An execution fragment $E$ is \emph{step contention-free for $T_k$} if the events of $E|k$ are contiguous in $E$, and $E$ is \emph{step contention-free} if $E$ is step contention-free for all transactions that participate in $E$.

\vspace{1mm}\noindent\textbf{TM correctness.}
A \emph{history} $H$ \emph{exported} by an execution fragment $E$, denoted $H_E$, is the subsequence of $E$ consisting of only the 
invocation and response events of t-operations.
Two histories $H$ and $H'$ are \emph{equivalent} if $\txns(H) = \txns(H')$
and for every transaction $T_k \in \txns(H)$, $H|k=H'|k$.
We say that two execution fragments $E$ and $E'$ are \emph{similar} if $H$ and $H'$ are equivalent, where $H$ (and resp. $H'$)
is the history exported by $E$ (and resp. $E'$).
For any two transactions $T_k,T_m\in \txns(E)$, we say that $T_k$ \emph{precedes}
$T_m$ in the \emph{real-time order} of $E$ ($T_k \prec_E^{RT} T_m$)
if $T_k$ is t-complete in $E$ and
the last event of $T_k$ precedes the first event of $T_m$ in $E$.
If neither $T_k$ precedes $T_m$ nor $T_m$ precedes $T_k$ in real-time order,
then $T_k$ and $T_m$ are \emph{concurrent} in $E$.
An execution $E$ is \emph{sequential} if every invocation of
a t-operation is either the last event in $H$ or
is immediately followed by a matching response, where $H$ is the history exported by $E$.
An execution $E$ is \emph{t-sequential} if there are no concurrent
transactions in $E$.

We say that $\Read_k(X)$ is \emph{legal} in a t-sequential execution $E$ if it returns the
latest written value of $X$, and $E$ is \emph{legal}
if every $\Read_k(X)$ in $H$ that does not return $A_k$ is legal in $E$.
Informally, a history $H$ is \emph{opaque}
if there exists a legal t-sequential history $S$ equivalent to $H$ that respects the real-time order of transactions
in $H$~\cite{tm-book}.
A weaker condition called \emph{strict serializability} ensures opacity only with respect to committed transactions.
Formal definitions are delegated to Appendix~\ref{app:upper}.

\vspace{1mm}\noindent\textbf{TM-liveness.}
A liveness property  specifies the conditions under which a
t-operation must return.
A TM implementation provides \emph{wait-free (WF)} TM-liveness
if it ensures that every t-operation returns in a finite number of its steps. 
A weaker property of  \emph{obstruction-freedom (OF)} 
ensures that every operation running step contention-free returns in
a finite number of its own steps.
The weakest property we consider here is \emph{sequential} TM-liveness
that only guarantees that t-operations running in the absence of
concurrent transactions returns in a finite number of its steps.  
%

\section{Hybrid Transactional Memory (HyTM)}
\label{sec:hytm}

\vspace{1mm}\noindent\textbf{Direct accesses and cached accesses.}
We  now describe the operation of a \emph{Hybrid Transactional Memoryг
  (HyTM)} implementation.
%
In our model, every base object can be accessed with two kinds of
primitives, \emph{direct} and \emph{cached}.

In a direct access, the rmw primitive operates on the memory state:
the direct-access event atomically reads the value of the object in
the shared memory and, if necessary, modifies it.

In a cached access performed by a process $i$, the rmw primitive operates on the \emph{cached}
state recorded in process $i$'s \emph{tracking set} $\tau_i$. 
One can think of $\tau_i$ as the \emph{L1 cache} of process $i$.
A a \emph{hardware transaction}  is a series of cached rmw primitives performed on $\tau_i$ followed by
a \emph{cache-commit} primitive. 
 
More precisely, $\tau_i$ is a set of triples $(b, v, m)$ where $b$ is a base object identifier, $v$ is a value, 
and $m \in \{\shared, \exclusive\}$ is an access \emph{mode}. 
The triple $(b, v, m)$ is added to the tracking set when $i$ performs a cached
rmw access of $b$, where $m$ is set to $\exclusive$ if the access is
nontrivial, and to $\shared$ otherwise.  
We assume that there exists some constant $\TS$ (representing the size of the L1 cache)
such that the condition $|\tau_i| \leq \TS$ must always hold; this
condition will be enforced by our model.
A base object $b$ is \emph{present} in $\tau_i$ with mode $m$ if $\exists v, (b,v,m) \in \tau_i$.

A trivial (resp.\ nontrivial) 
cached primitive $\langle g,h \rangle$ applied to $b$ 
by process $i$ first checks the condition $|\tau_i|=\TS$ and if so, it
sets $\tau_i=\emptyset$ and immediately returns $\bot$ (we call this event a
\emph{capacity abort}). 
We assume that $\TS$ is large enough so that no transaction 
with data set of size $1$ can incur a capacity abort.
If the transaction does not incur a capacity abort, the process checks whether $b$ is present in exclusive
(resp.\ any) mode in $\tau_j$ 
for any $j\neq i$. If so, $\tau_i$ is set to $\emptyset$ and the
primitive returns $\bot$. 
Otherwise, the triple $(b, v, \shared)$ (resp. $(b, g(v), \exclusive)$)
is added to $\tau_i$,  where $v$ is the most recent cached value of $b$ in $\tau_i$
(in case $b$ was previously accessed by $i$ within the current
hardware transaction) or the value of $b$ in the current
memory configuration, and finally $h(v)$ is returned.

A tracking set can be \emph{invalidated} by a concurrent process: 
if, in a configuration $C$ where  $(b,v,\exclusive)\in\tau_i$
(resp.\ $(b,v,\shared)\in\tau_i)$,  a process $j\neq i$ applies any primitive 
(resp.\ any \emph{nontrivial} primitive) to $b$, then $\tau_i$ becomes
\emph{invalid} and any subsequent cached primitive invoked by $i$
sets $\tau_i$ to $\emptyset$ and returns $\bot$. We refer to this event as a \emph{tracking set abort}.

Finally, the \emph{cache-commit} primitive issued by process $i$ with
a valid $\tau_i$ does the following: for each base object $b$ such that $(b,v,\exclusive) \in \tau_i$, the value of $b$ in $C$ is updated to $v$. 
Finally, $\tau_i$ is set to $\emptyset$ and the primitive 
returns $\textit{commit}$. 

Note that HTM may also abort spuriously, or because of unsupported operations~\cite{Rei12}. 
The first cause can be modelled probabilistically in the above
framework, which would not however significantly affect our claims and proofs, except for a more cumbersome presentation. 
Also, our lower bounds are based exclusively on executions containing t-reads and t-writes. 
Therefore, in the following, we only consider contention and capacity aborts.  

\vspace{1mm}\noindent\textbf{Slow-path and fast-path transactions.}
In the following, we partition HyTM transactions into \emph{fast-path
  transactions} and \emph{slow-path transactions}.
Practically,  two separate algorithms (fast-path one and slow-path one) 
are provided for each t-operation. 

A slow-path transaction models a regular software transaction.
An event of a slow-path transaction is either an invocation or response of a t-operation, or
a  rmw primitive on a base object. 

A fast-path transaction essentially encapsulates a hardware transaction. 
An event of a fast-path transaction is either an invocation or response of a t-operation, 
a cached primitive on a base object, or a \emph{cache-commit}:
\textit{t-read} and \emph{t-write} are only allowed to contain cached
primitives, and \textit{tryC} consists of invoking \emph{cache-commit}.  
Furthermore, we assume that a fast-path transaction $T_k$ returns $A_k$
as soon an underlying cached primitive or \emph{cache-commit} returns $\bot$. 
Figure~\ref{fig:tracking-set} depicts such a scenario illustrating a tracking set abort: 
fast-path transaction $T_2$ executed by process $p_2$
accesses a base object $b$ in shared (and resp. exclusive) mode and it is added to its tracking set $\tau_2$. 
Immediately after the access of $b$ by $T_2$, a concurrent transaction $T_1$ applies a nontrivial primitive to $b$ 
(and resp. accesses $b$). Thus, the tracking of $p_2$ is invalidated and $T_2$ must be aborted in any extension of this execution.
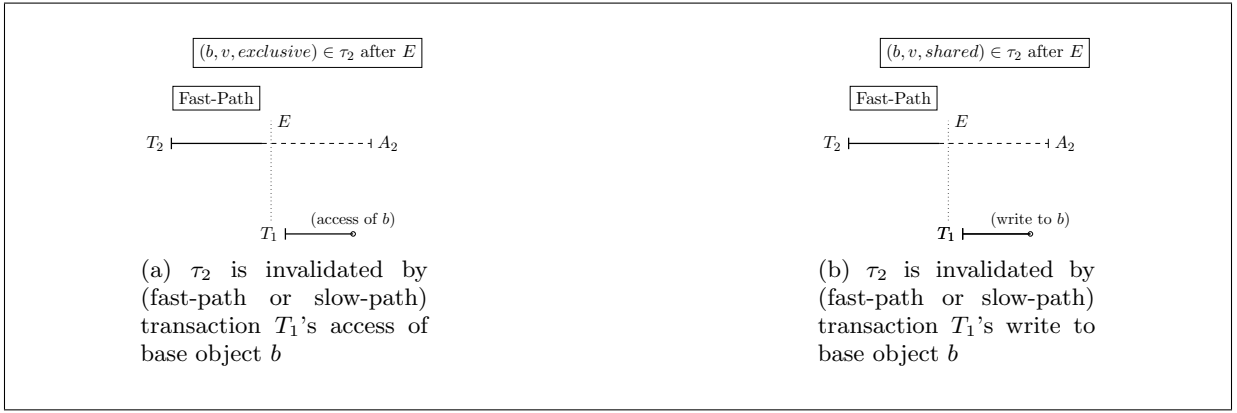
\begin{figure*}[t]
\begin{center}
	\subfloat[$\tau_2$ is invalidated by (fast-path or slow-path) transaction $T_1$'s access of base object $b$ \label{sfig:inv-1}]{\scalebox{0.6}[0.6]{\begin{tikzpicture}
\node (e) at (13,-2) [] {};

\node[draw,align=left] at (10,1) {Fast-Path};
\draw (e) node [above] {\small {(access of $b$)}};

\begin{scope}   
\draw [|-,thick] (9,0) node[left] {$T_2$} to (11,0);
\draw [-|,dashed] (11,0)  to (13.4,0) node[right] {$A_2$} ;
\draw [|-,thick] (11.5,-2) node[left] {$T_1$} to (13,-2);

\draw [-,dotted] (11.2,-1.7)  to (11.2,0.5) node[right] {$E$};
\end{scope}
\draw  (13,-2) circle [fill, radius=0.05]  (13,-2);

\node[draw,align=right] at (12,2) { $(b,v,exclusive) \in \tau_2$ after $E$};

\end{tikzpicture}}}
        \hspace{50mm}
	\subfloat[$\tau_2$ is invalidated by (fast-path or slow-path) transaction $T_1$'s write to base object $b$ \label{sfig:inv-2}]{\scalebox{0.6}[0.6]{\begin{tikzpicture}
\node (e) at (13+6,-2) [] {};

\node[draw,align=left] at (10+6,1) {Fast-Path};
\draw (e) node [above] {\small {(write to $b$)}};

\begin{scope}   
\draw [|-,thick] (9+6,0) node[left] {$T_2$} to (11+6,0);
\draw [-|,dashed] (11+6,0)  to (13.4+6,0) node[right] {$A_2$} ;
\draw [|-,thick] (11.5+6,-2) node[left] {$T_1$} to (13+6,-2);
\draw [|-,thick] (11.5+6,-2) node[left] {$T_1$} to (13+6,-2);

\draw [-,dotted] (11.2+6,-1.7)  to (11.2+6,0.5) node[right] {$E$};
\end{scope}
\draw  (13+6,-2) circle [fill, radius=0.05]  (13+6,-2);

\node[draw,align=right] at (10+8,2) { $(b,v,shared) \in \tau_2$ after $E$};

\end{tikzpicture}}}
	
\end{center}

\caption{Tracking set aborts in fast-path transactions;
we denote a fast-path (and resp. slow-path) transaction by $F$ (and resp. $S$)
\label{fig:tracking-set}} 
\end{figure*}
%

We provide two key observations on this model regarding the interactions of non-committed fast path transactions 
with other transactions. 
Let $E$ be any execution of a HyTM implementation $\mathcal{M}$ in
which a fast-path transaction $T_k$ is either
t-incomplete or aborted. 
Then the sequence of events $E'$ derived by removing all events of $E|k$
from $E$ is an execution  $\mathcal{M}$. Moreover: 
%
\begin{observation} 
\label{ob:one}
To every slow-path transaction $T_m \in \ms{txns}(E)$, $E$ is indistinguishable 
from $E'$. 
\end{observation}
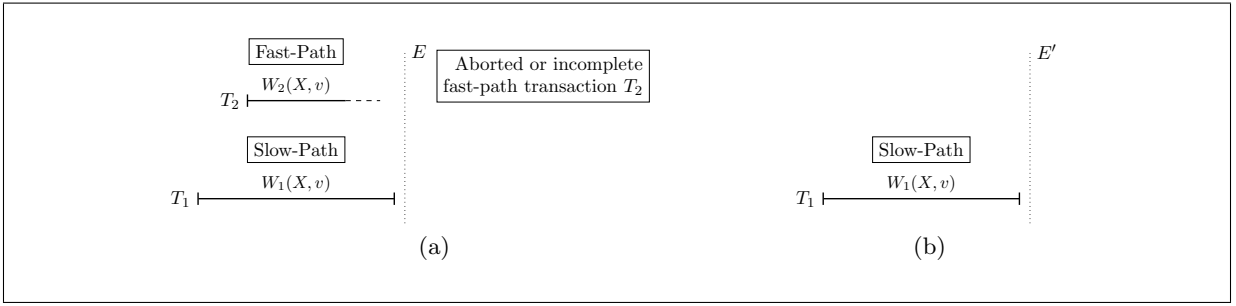
\begin{figure*}[t]
\begin{center}
	\subfloat[\label{sfig:ob-01}]{\scalebox{0.65}[0.65]{\begin{tikzpicture}
\node (w2) at (10,0) [] {};
\node (w1) at (10,-2) [] {};
\node (w3) at (18,-2) [] {};

\draw (w2) node [above] {\small {$W_2(X,v)$}};

\draw (w1) node [above] {\small {$W_1(X,v)$}};

\node[draw,align=left] at (10,1) {Fast-Path};
\node[draw,align=left] at (10,-1) {Slow-Path};

\begin{scope}   
\draw [|-,thick] (9,0) node[left] {$T_2$} to (11,0);
\draw [-,dashed] (11,0) to (11.7,0);
\draw [|-|,thick] (8,-2) node[left] {$T_1$} to (12,-2);
\draw [-,dotted] (12.2,-2.5)  to (12.2,1) node[right] {$E$};
\end{scope}
\node[draw,align=right] at (15,.5) {Aborted or incomplete\\ fast-path transaction $T_2$};

\end{tikzpicture}}}
	\hspace{10mm}
	\subfloat[\label{sfig:ob-02}]{\scalebox{0.65}[0.65]{\begin{tikzpicture}

\node (w1) at (11,-2) [] {};

\draw (w1) node [above] {\small {$W_1(X,v)$}};

\node[draw,align=left] at (11,-1) {Slow-Path};

\begin{scope}   
\draw [|-|,thick] (9,-2) node[left] {$T_1$} to (13,-2);
\draw [-,dotted] (13.2,-2.5)  to (13.2,1) node[right] {$E'$};
\end{scope}
\end{tikzpicture}}}
	 
\end{center}
\caption{
 \label{fig:ob1}
 Execution $E$ in Figure~\ref{sfig:ob-01} is indistinguishable
to $T_1$ from the execution $E'$ in Figure~\ref{sfig:ob-02}}
\end{figure*}
\begin{observation} 
\label{ob:two}
If a fast-path transaction $T_m\in \ms{txns}(E) \setminus \{T_k\}$ does not incur a tracking set abort in $E$, 
then $E$ is indistinguishable to $T_m$ from $E'$.
\end{observation}
Intuitively, these observations say that fast-path transactions which are not yet committed are 
invisible to slow-path transactions, and can communicate with other
fast-path transactions only by incurring their tracking-set aborts.
Figure~\ref{fig:ob1} illustrates Observation~\ref{ob:one}: a fast-path transaction $T_2$ is concurrent to a slow-path transaction
$T_1$ in an execution $E$. Since $T_2$ is t-incomplete or aborted in this execution, $E$ is indistinguishable to $T_1$
from an execution $E'$ derived by removing all events of $T_2$ from $E$.
Analogously, to illustrate Observation~\ref{ob:two}, if $T_1$ is a fast-path transaction that does not incur a tracking set abort in $E$, then
$E$ is indistinguishable to $T_1$ from $E'$.
%
\section{Instrumentation}
\label{sec:ins}
Now we define the notion of \emph{code instrumentation} in fast-path transactions.
%

An execution $E$ of a HyTM $\mathcal{M}$ \emph{appears t-sequential}
to a transaction $T_{k} \in \ms{txns}(E)$ if there exists an execution
$E'$ of $\mathcal{M}$ such that:
\begin{itemize}
\item $\ms{txns}(E')\subseteq \ms{txns}(E)\setminus \{T_k\}$ and the configuration after $E'$ is t-quiescent,
\item every transaction $T_m\in \ms{txns}(E)$
that precedes $T_k$ in real-time order is included in $E'$ such that $E|m=E'|m$,
\item for every transaction $T_m\in \ms{txns}(E')$, 
$\Rset_{E'}(T_m)\subseteq \Rset_E(T_m)$ and $\Wset_{E'}(T_m)\subseteq \Wset_E(T_m)$, and
\item $E'\cdot E|k$ is an execution of $\mathcal{M}$.
\end{itemize}
%
%
%

\begin{definition}[Data and metadata base objects] 
\label{def:metadata}
Let $\mathcal{X}$ be the set of t-objects operated by a HyTM implementation $\mathcal{M}$. 
Now we partition the set of base objects used by $\mathcal{M}$ into a set $\mathbb{D}$ of \emph{data} objects and 
a set $\mathbb{M}$ of \emph{metadata} objects ($\mathbb{D}\cap \mathbb{M} = \emptyset$). We further partition
$\mathbb{D}$ into sets $\mathbb{D}_X$ associated with each t-object $X
\in \mathcal{X}$:  $\mathbb{D} = \bigcup\limits_{X\in\mathcal{X}} \mathbb{D}_X$,
for all $X\neq Y$ in $\X$, $\mathbb{D}_X \cap \mathbb{D}_Y =
\emptyset$,
such that:
%
\begin{enumerate}
\item In every execution $E$, each fast-path transaction $T_k \in \ms{txns}(E)$ only 
accesses base objects in $\bigcup\limits_{X\in DSet(T_k)}
\mathbb{D}_X$ or $\mathbb{M}$.
\item
Let $E\cdot\rho$ and $E\cdot E'\cdot\rho'$ be two t-complete
executions, such that $E$ and $E\cdot E'$ are t-complete, 
$\rho$ and $\rho'$ are complete executions of a transaction
$T_k\notin\ms{txns}(E\cdot E')$, $H_{\rho}=H_{\rho'}$, and $\forall
T_m\in \ms{txns}(E')$, $\Dset(T_m)\cap \Dset(T_k)=\emptyset$. 
Then the states of the base objects $\bigcup\limits_{X\in DSet(T_k)} \mathbb{D}_X$ 
in the configuration after $E\cdot \rho$ and $E\cdot E' \cdot {\rho'}$
are the same.    


\item 
Let execution $E$ appear t-sequential to a transaction $T_k$ and let
the enabled event $e$ of $T_k$ after $E$ be a primitive on a base
object $b\in \mathbb{D}$. Then, unless $e$ returns $\bot$, $E\cdot e$
also appears t-sequential to $T_k$. 
 
%
\end{enumerate}
\end{definition}
Intuitively, the first condition says that a transaction is only
allowed to access data objects based on its data set. The second
condition says that transactions with disjoint data sets can
communicate only via metadata objects.
Finally, the last condition means that base objects in
$\mathbb{D}$ may only contain the ``values'' of t-objects, and cannot be used to detect concurrent transactions. Note that our results will lower bound the number of metadata objects that must be accessed under particular assumptions, thus from a cost perspective, $\mathbb{D}$ should be made as large as possible.

All HyTM proposals we aware of, such as 
\emph{HybridNOrec}~\cite{hybridnorec,riegel-thesis}, \emph{PhTM}~\cite{phasedtm} and others~\cite{damronhytm, kumarhytm},
conform to our definition of instrumentation in fast-path transactions.
For instance, HybridNOrec~\cite{hybridnorec,riegel-thesis} employs a distinct base object in $\mathbb{D}$ for each
t-object and a global \emph{sequence lock} as the metadata that is accessed by fast-path transactions to detect
concurrency with slow-path transactions.
Similarly, the HyTM implementation by \emph{Damron et al.}~\cite{damronhytm} also 
associates a distinct base object in $\mathbb{D}$ for each
t-object and additionally, a \emph{transaction header} and \emph{ownership record} as metadata base objects.

%
%
\begin{definition}[Uninstrumented HyTMs]
\label{def:ins}
A HyTM implementation $\mathcal{M}$ provides \emph{uninstrumented writes (resp.\ reads)} 
if in every execution $E$ of $\mathcal{M}$, for every write-only (resp.\ read-only) 
fast-path
 transaction $T_k$, 
all primitives in $E|k$ are performed on base objects in $\mathbb{D}$.
A HyTM is uninstrumented if both its reads and writes are uninstrumented. 
\end{definition}
%
%
\begin{observation}
\label{ob:ins}
Consider any execution $E$ of a HyTM implementation $\mathcal{M}$ which provides uninstrumented reads (resp. writes). 
For any fast-path read-only (resp.\ write-only) transaction $T_k \not\in \ms{txns}(E)$, 
that runs step-contention free after $E$, 
the execution $E$ appears t-sequential to $T_k$.
\end{observation}

\fi
%
\begin{figure*}[t]
\begin{center}
	\subfloat[$T_y$ must return the new value\label{sfig:inv-0}]{\scalebox{0.65}[0.65]{\begin{tikzpicture}
\node (r1) at (1,0) [] {};
\node (w1) at (3,0) [] {};
\node (c1) at (5,0) [] {};

\node (e) at (7.5,0) [] {};

\node (r3) at (10,0) [] {};

\draw (r1) node [above] {\small {$R_0(Z)\rightarrow v$}};
\draw (w1) node [above] {\small {$W_0(X,nv)$}};
\draw (c1) node [above] {\small {$\TryC_0$}};

\draw (e) node [above] {\tiny {(event of $T_0$)}};
\draw (e) node [below] {\small {$e$}};

\draw (r3) node [above] {\small {$R_y(Y)\rightarrow nv$}};
\draw (r3) node [below] {\tiny {returns new value}};

\node[draw,align=left] at (2.5,1) {S};
\node[draw,align=left] at (10,1) {F};
\begin{scope}   
\draw [|-|,thick] (0,0) node[left] {$T_0$} to (2,0);
\draw [-|,thick] (2,0) node[left] {} to (4,0);
\draw [-,thick] (4,0) node[left] {} to (5.5,0);
\draw  (7.5,0) circle [fill, radius=0.05]  (7.5,0);
\draw [|-|,thick] (9,0) node[left] {$T_y$} to (11,0);
\end{scope}
\end{tikzpicture}}}
	\\
	\subfloat[Since $T_z$ is uninstrumented, by Observation~\ref{ob:ins} and sequential TM-progress, $T_z$ must commit\label{sfig:inv-1}]{\scalebox{0.65}[0.65]{\begin{tikzpicture}
\node (r1) at (1,0) [] {};
\node (w1) at (3,0) [] {};
\node (c1) at (5,0) [] {};

\node (w3) at (8.5,0) [] {};

\draw (r1) node [above] {\small {$R_0(Z)\rightarrow v$}};
\draw (w1) node [above] {\small {$W_0(X,nv)$}};
\draw (c1) node [above] {\small {$\TryC_0$}};

\draw (w3) node [above] {\small {$W_z(Z,nv)$}};
\draw (w3) node [below] {\tiny {write new value}};

\node[draw,align=left] at (2.5,1) {S};
\node[draw,align=left] at (8.5,1) {F};

\begin{scope}   
\draw [|-|,thick] (0,0) node[left] {$T_0$} to (2,0);
\draw [-|,thick] (2,0) node[left] {} to (4,0);
\draw [-,thick] (4,0) node[left] {} to (5.5,0);
\draw [|-|,thick] (7.5,0) node[left] {$T_z$} to (9.5,0);
\end{scope}
\end{tikzpicture}}}
        \\
        \vspace{2mm}
	\subfloat[Since $T_x$ does not access any metadata, by Observation~\ref{ob:ins}, it cannot abort and must 
	return the initial value value of $X$\label{sfig:inv-2}]{\scalebox{0.65}[0.65]{\begin{tikzpicture}
\node (r1) at (1,0) [] {};
\node (w1) at (3,0) [] {};
\node (c1) at (5,0) [] {};

\node (r2) at (12,0) [] {};

\node (w3) at (8.5,0) [] {};

\draw (r1) node [above] {\small {$R_0(Z)\rightarrow v$}};
\draw (w1) node [above] {\small {$W_0(X,nv)$}};
\draw (c1) node [above] {\small {$\TryC_0$}};

\draw (r2) node [above] {\small {$R_x(X)\rightarrow v$}};
\draw (r2) node [below] {\tiny {returns initial value}};

\draw (w3) node [above] {\small {$W_z(Z,nv)$}};
\draw (w3) node [below] {\tiny {write new value}};

\node[draw,align=left] at (2.5,1) {S};
\node[draw,align=left] at (8.5,1) {F};
\node[draw,align=left] at (12,1) {F};

\begin{scope}   
\draw [|-|,thick] (0,0) node[left] {$T_0$} to (2,0);
\draw [-|,thick] (2,0) node[left] {} to (4,0);
\draw [-,thick] (4,0) node[left] {} to (5.5,0);
\draw [|-|,thick] (7.5,0) node[left] {$T_z$} to (9.5,0);
\draw [|-|,thick] (11,0) node[left] {$T_x$} to (13,0);
\end{scope}
\end{tikzpicture}}}
	\\
	\vspace{2mm}
	\subfloat[$T_y$ does not contend with $T_x$ or $T_z$ on any base object \label{sfig:inv-3}]{\scalebox{0.65}[0.65]{\begin{tikzpicture}
\node (r1) at (1,0) [] {};
\node (w1) at (3,0) [] {};
\node (c1) at (5,0) [] {};

\node (r2) at (12,0) [] {};

\node (e) at (15,0) [] {};

\node (r3) at (18.5,0) [] {};

\node (w3) at (8.5,0) [] {};

\draw (r1) node [above] {\small {$R_0(Z)\rightarrow v$}};
\draw (w1) node [above] {\small {$W_0(X,nv)$}};
\draw (c1) node [above] {\small {$\TryC_0$}};

\draw (r2) node [above] {\small {$R_x(X)\rightarrow v$}};
\draw (r2) node [below] {\tiny {returns initial value}};

\draw (e) node [above] {\tiny {(event of $T_0$)}};
\draw (e) node [below] {\small {$e$}};

\draw (r3) node [above] {\small {$R_y(Y)\rightarrow nv$}};
\draw (r3) node [below] {\tiny {returns new value}};

\draw (w3) node [above] {\small {$W_z(Z,nv)$}};
\draw (w3) node [below] {\tiny {write new value}};

\node[draw,align=left] at (2.5,1) {S};
\node[draw,align=left] at (8.5,1) {F};
\node[draw,align=left] at (18.5,1) {F};
\node[draw,align=left] at (12,1) {F};

\begin{scope}   
\draw [|-|,thick] (0,0) node[left] {$T_0$} to (2,0);
\draw [-|,thick] (2,0) node[left] {} to (4,0);
\draw [-,thick] (4,0) node[left] {} to (5.5,0);
\draw [|-|,thick] (7.5,0) node[left] {$T_z$} to (9.5,0);
\draw [|-|,thick] (11,0) node[left] {$T_x$} to (13,0);
\draw  (15,0) circle [fill, radius=0.05]  (15,0);
\draw [|-|,thick] (17.5,0) node[left] {$T_y$} to (19.5,0);
\end{scope}
\end{tikzpicture}}}
	\caption{Executions in the proof of Theorem~\ref{instrumentation}; execution in \ref{sfig:inv-3} is not strictly serializable
          \label{fig:indis}} 
\end{center}
\end{figure*}
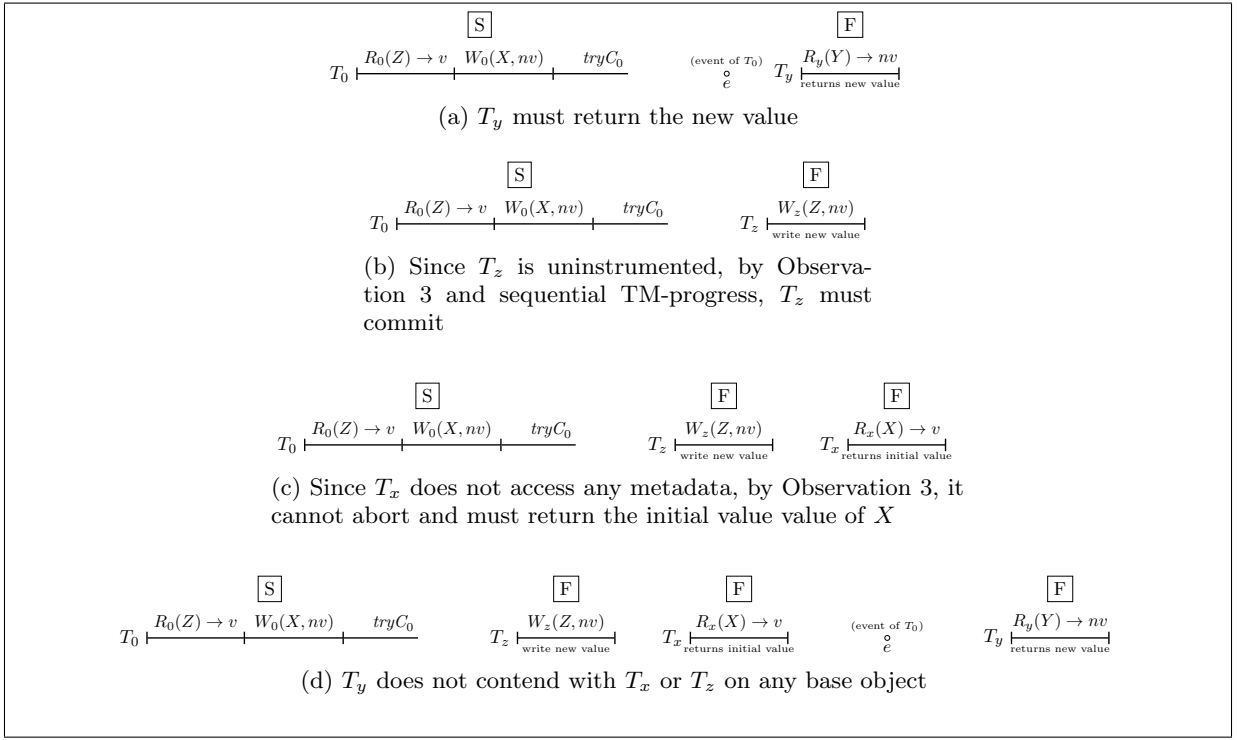
\section{Impossibility of uninstrumented HyTMs}
\label{sec:main}
In this section, we show that any strictly serializable HyTM must be
instrumented, even under a very weak progress assumption by which  
a transaction is guaranteed to commit only when run t-sequentially:

\begin{definition}[Sequential TM-progress]
\label{def:seqtmprogress}
A HyTM implementation $\mathcal{M}$ provides \emph{sequential TM-progress for fast-path transactions (and resp. slow-path)} if 
in every execution $E$ of $\mathcal{M}$, a fast-path 
(and resp. slow-path) 
transaction $T_k$ returns $A_k$ in $E$ only if 
$T_k$ incurs a capacity abort or $T_k$ is concurrent to another transaction. 
We say that $\mathcal{M}$ provides sequential TM-progress if it provides sequential TM-progress for fast-path and slow-path
transactions. 
\end{definition}

\begin{theorem}
\label{instrumentation}
There does not exist a strictly serializable uninstrumented HyTM
implementation 
that ensures sequential TM-progress and TM-liveness.
\end{theorem}
\begin{proof}
Suppose by contradiction that such a HyTM $\M$ exists. 
For simplicity, assume that $v$ is the initial value of t-objects $X$, $Y$ and $Z$.
Let $E$ be the t-complete step contention-free execution of a slow-path transaction $T_0$
that performs $\Read_0(Z) \rightarrow v$,  $\Write_0(X, nv)$,
$\Write_0(Y, nv)$ ($nv\neq v$), and commits. 
Such an execution exists since $\M$ ensures sequential TM-progress. 

By Observation~\ref{ob:ins}, any transaction that 
runs step contention-free starting from a prefix of $E$ must return a
non-abort value.
Since any such transaction reading $X$ or $Y$ must return $v$ 
when it starts from the
empty prefix of $E$ and $nv$ when it starts from $E$.

Thus, there exists  $E'$, the longest prefix of $E$ that cannot be extended with the 
t-complete step contention-free execution of a \emph{fast-path} transaction 
reading $X$ or $Y$ and returning $nv$.
Let $e$ is the enabled event of $T_0$ in the configuration after $E'$.
Without loss of generality, suppose that there exists an execution
$E'\cdot e\cdot E_y$ where $E_y$ is the t-complete step contention-free
execution fragment of some fast-path transaction $T_y$ that reads $Y$ is returns $nv$ (Figure~\ref{sfig:inv-0}). 
%
%
\begin{claim}\label{claim:concat}
$\M$ has an execution $E' \cdot E_z \cdot  E_x$, where
\begin{itemize}
\item
$E_z$ is the t-complete step contention-free execution fragment of a fast-path transaction $T_{z}$ that 
writes $nv \neq v$ to $Z$ and commits
\item
$E_x$ is the t-complete step contention-free execution fragment of a fast-path transaction $T_x$ that performs
a single t-read $\Read_x(X) \rightarrow v$ and commits.
\end{itemize}
\end{claim}
\begin{proof}
%
By Observation~\ref{ob:ins}, the extension of $E'$ in which $T_z$ writes to
$Z$ and tries to commit appears t-sequential to $T_z$.
By sequential TM-progress, $T_z$ complets the write and commits. 
Let  $E' \cdot E_z$ (Figure~\ref{sfig:inv-1})
be the resulting execution of $\M$.

Similarly, the extension of $E'$ in which $T_x$ reads $X$ 
and tries to commit appears t-sequential to $T_x$.
By sequential TM-progress, $T_x$ commits and let  $E' \cdot E_x$
be the resulting execution of $\M$.
By the definition of $E'$, 
$\Read_x(X)$ must return $v$ in $E' \cdot E_x$.

Since $\M$ is uninstrumented and the data sets of $T_x$
and $T_z$ are disjoint, 
the sets of base objects accessed in the execution fragments $E_x$ and
$E_y$ are also disjoint.
Thus, $E' \cdot E_z \cdot E_x $
is indistinguishable to $T_x$ from the execution $E' \cdot E_x$, which
implies that $E' \cdot E_z \cdot E_x$ is an execution of $\M$ (Figure~\ref{sfig:inv-2}).
\end{proof} 
Finally, we prove that the sequence of events, ${E' \cdot E_z \cdot E_x \cdot e \cdot E_y}$ is an execution of $\M$.

Since the transactions $T_x$, $T_y$, $T_z$ have
pairwise disjoint data sets in ${E' \cdot E_z \cdot E_x \cdot e \cdot E_y}$,
no base object accessed ib $E_y$ can be accessed in $E_x$ and $E_z$.
The read operation on $X$ performed by $T_y$ in $E'\cdot e\cdot E_y$
returns $nv$ and, by the definition of $E'$ and $e$, $T_y$ must have accessed
the base object $b$ modified in the event $e$ by $T_0$.
Thus, $b$ is not accessed in $E_x$ and $E_z$ and 
$E' \cdot E_z \cdot E_x \cdot e$ is an execution of $\M$.
Summing up, $E' \cdot E_z \cdot E_x \cdot e \cdot E_y$
is indistinguishable to $T_y$ from
$E'  \cdot e \cdot E_y$, which implies that 
$E' \cdot E_z \cdot E_x \cdot e \cdot E_y$ is an execution of $\M$
(Figure~\ref{sfig:inv-3}).

But the resulting execution is not strictly serializable.
Indeed, suppose that a serialization exists.  
As the value written by $T_0$ is returned by a committed transaction
$T_y$, $T_0$ must be committed and precede $T_y$ in the serialization.
Since $T_x$ returns the initial value of $X$, $T_x$ must precede
$T_0$. 
Since $T_0$ reads the initial value of $Z$, $T_0$ must
precede $T_z$.
Finally, $T_z$ must precede $T_x$ to respect the real-time order. 
The cycle in the serialization establishes a contradiction.
\end{proof}
%
%
%
%
%
%
%
%

\section{Providing concurrency in HyTM}
\label{sec:main2}
In this section, we show that giving HyTM the ability to run and commit
transactions in parallel brings 
considerable 
instrumentation costs.   
We focus on a 
natural
progress condition called
progressiveness~\cite{GK08-opacity,GK09-progressiveness,tm-theory} that allows a
transaction to abort only if it experiences a read-write or write-write
conflict with a concurrent transaction: 
\begin{definition}[Progressiveness]
\label{def:prog}
Transactions $T_i$ and $T_j$ \emph{conflict} in an execution $E$ 
on a t-object $X$ if
$X\in\Dset(T_i)\cap\Dset(T_j)$ and $X\in\Wset(T_i)\cup\Wset(T_j)$.
A HyTM implementation $\M$ 
is \emph{fast-path} (resp. \emph{slow-path}) \emph{progressive} 
if in every execution $E$ of $\M$ and for every fast-path (and resp. slow-path) transaction 
$T_i $ that aborts in $E$, 
either $A_i$ is a capacity abort or  $T_i$ conflicts with some transaction $T_j$ that is concurrent to $T_i$ in $E$.  
We say $\M$ is \emph{progressive} if it is both fast-path and slow-path progressive.
%
\end{definition}
%
%
\subsection{A linear lower bound on instrumentation}
%
We show that for every opaque fast-path progressive HyTM that provides
obstruction-free TM-liveness, an arbitrarily long read-only
transaction might access a number of distinct
metadata base objects that is linear in the size of its read set or
experience a capacity abort.

The following auxiliary results will be crucial in proving our lower
bound.
We observe first that a fast path transaction in a progressive HyTM can 
contend on a base object only with a non-conflicting transaction.
%
\ignore{
\begin{lemma}
\label{lm:pgone}
%
Let $E$ be any execution of $\mathcal{M}$.
Then, for any two transactions $\{T_1, T_2\}\in \ms{txns}(E)$ that do not conflict in $E$, 
if at least one of $T_1$ or $T_2$ is a fast-path transaction, then
$T_1$ and $T_2$ do not concurrently contend on any base object in $E$.
\end{lemma}
\begin{proof}
Suppose by contradiction that there exists an execution $E$ of $\mathcal{M}$ and 
$\{T_1, T_2\}\in \ms{txns}(E)$,  such that $\Dset(T_1)\cap \Dset(T_2)=\emptyset$, $T_1$
is a fast-path transaction, and $T_1$, $T_2$ concurrently contend on some base object $b$ in $E$.
Let $e_1$ (and resp. $e_2$) be the enabled event of transaction $T_1$ (and resp. $T_2$) after $E$.
Let $p_1$ (and resp. $p_2$) be the process performing transaction $T_1$ (and resp. $T_2$).
At least one event of $e_1$ and $e_2$ must be nontrivial.  

Consider the execution $E\cdot e_1 \cdot e_2'$ where $e_2'$ is the
event of $p_2$ in which it applies the primitive of $e_2$ to the
configuration after $E\cdot e_1$. 
After $E\cdot e_1$, $b$ is contained in the tracking set of process
$p_1$.
If $b$ is contained in $\tau_1$ in the shared mode, then $e_2'$ is a
nontrivial primitive on $b$, which invalidates $\tau_1$ in $E\cdot
e_1\cdot e_2'$.
If  $b$ is contained in $\tau_1$ in the exclusive mode, then any
subsequent access of $b$ invalidates $\tau_1$ in  $E\cdot
e_1\cdot e_2'$.
In both cases,   $\tau_1$ is invalidated and $T_1$ incurs a tracking set abort.
Thus, transaction $T_1$ must return $A_1$ in any extension of $E\cdot
e_1\cdot e_2$---a contradiction
to the assumption that $\mathcal{M}$ is progressive.
\end{proof}
}
%
\begin{lemma}
\label{lm:pgtwo}
Let $\mathcal{M}$ be any fast-path progressive HyTM implementation.
Let $E\cdot E_1 \cdot E_2$ be an execution of $\mathcal{M}$ where
$E_1$ (and resp. $E_2$) is the step contention-free
execution fragment of transaction $T_1 \not\in \ms{txns}(E)$ (and
resp. $T_2 \not\in \ms{txns}(E)$),
$T_1$ (and resp. $T_2$) does not conflict with any transaction in $E\cdot E_1 \cdot E_2$, and
at least one of $T_1$ or $T_2$ is a fast-path transaction. 
Then, $T_1$ and $T_2$ do not contend on any base object in $E\cdot E_1 \cdot E_2$.
\end{lemma}
\begin{proof}
Suppose, by contradiction that $T_1$ or $T_2$ 
contend on the same base object in $E\cdot E_1\cdot E_2$.

If in $E_1$, $T_1$ performs a nontrivial event on a base object on which they contend, let $e_1$ be the last
event in $E_1$ in which $T_1$ performs such an event to some base object $b$ and $e_2$, the first event
in $E_2$ that accesses $b$.
Otherwise, $T_1$
only performs trivial events in $E_1$ to base objects on which it contends with $T_2$ in $E\cdot E_1\cdot E_2$:
let $e_2$ be the first event in $E_2$ in which $E_2$ performs a nontrivial event to some base object $b$
on which they contend and $e_1$, the last event of $E_1$ in $T_1$ that accesses $b$.

Let $E_1'$ (and resp. $E_2'$) be the longest prefix of $E_1$ (and resp. $E_2$) that does not include
$e_1$ (and resp. $e_2$).
Since before accessing $b$, the execution is step contention-free for $T_1$, $E \cdot
E_1'\cdot E_2'$ is an execution of $\mathcal{M}$.
By construction, $T_1$ and $T_2$ do not conflict in $E \cdot E_1'\cdot E_2'$.
Moreover, $E\cdot E_1 \cdot E_2'$ is indistinguishable to $T_2$ from $E\cdot E_1' \cdot E_2'$.
Hence, $T_1$ and
$T_2$ are poised to apply contending events $e_1$ and $e_2$ on $b$ in the execution
$\tilde E=E\cdot E_1' \cdot E_2'$.
Recall that at least one event of $e_1$ and $e_2$ must be nontrivial.  

Consider the execution $\tilde E\cdot e_1 \cdot e_2'$ where $e_2'$ is the
event of $p_2$ in which it applies the primitive of $e_2$ to the
configuration after $\tilde E \cdot e_1$. 
After $\tilde E\cdot e_1$, $b$ is contained in the tracking set of process
$p_1$.
If $b$ is contained in $\tau_1$ in the shared mode, then $e_2'$ is a
nontrivial primitive on $b$, which invalidates $\tau_1$ in $\tilde E\cdot
e_1\cdot e_2'$.
If  $b$ is contained in $\tau_1$ in the exclusive mode, then any
subsequent access of $b$ invalidates $\tau_1$ in  $\tilde E\cdot
e_1\cdot e_2'$.
In both cases,   $\tau_1$ is invalidated and $T_1$ incurs a tracking set abort.
Thus, transaction $T_1$ must return $A_1$ in any extension of $E\cdot
e_1\cdot e_2$---a contradiction
to the assumption that $\mathcal{M}$ is progressive.   
\end{proof}
Iterative application of Lemma~\ref{lm:pgtwo} implies the following:  
\begin{corollary}
\label{lm:pgthree}
Let $\mathcal{M}$ be any fast-path progressive HyTM implementation.
Let $E\cdot E_1 \cdots E_i \cdot E_{i+1} \cdots E_m$ be any execution of $\mathcal{M}$ where
$E_i$ is the step contention-free
execution fragment of transaction $T_i \not\in \ms{txns}(E)$,
for all $i\in \{1,\ldots , m\}$
and any two transactions in $E_1 \cdots E_m$ do not conflict.
For all $i,j=1,\ldots,m$, $i\neq j$, if $T_i$ is fast-path, then $T_i$
and $T_j$ do not contend on a
base object in  $E\cdot E_1 \cdots E_{m} \cdots E_m$   
\end{corollary}
\begin{proof}
Let $T_i$ be a fast-path transaction.
By Lemma~\ref{lm:pgtwo}, in $E\cdot E_1 \cdots E_{i} \cdots E_m$, 
$T_i$ does not contend with $T_{i-1}$ (if $i>1$) or
$T_{i+1}$ (if $i<m$) on any base object
and, thus, $E_i$ commutes with $E_{i-1}$ and $E_{i+1}$.  
Thus, 
$E\cdot E_1 \cdots E_{i-2}\cdot E_{i}\cdot E_{i-1}\cdot E_{i+1}\cdots E_m$  (if $i>1$) and 
$E\cdot E_1 \cdots E_{i-1}\cdot E_{i+1}\cdot E_{i}\cdot E_{i+2}\cdots E_m$  (if $i<m$) are executions of
$\M$.
By iteratively applying Lemma~\ref{lm:pgtwo}, we derive that $T_i$
does not contend with any $T_j$, $j\neq i$.
\end{proof}
%
%
%
%
Recall that execution fragments $E$ and $E'$ are called similar if
they export equivalent histories, \emph{i.e.}, no process can see the
difference between them by looking at the invocations and responses of
t-operations.   
We now use Corollary~\ref{lm:pgthree} to show that t-operations
only accessing data base objects cannot detect contention with
non-conflicting transactions.
\begin{lemma}
\label{lm:finallm}
Let $E$ be any t-complete execution of 
a progressive  HyTM implementation 
$\M$ that provides OF TM-liveness.
For any $m\in \mathbb{N}$, consider a set of $m$ executions of $\M$ of the form
$E\cdot E_{i} \cdot \gamma_i \cdot \rho_i$ where
$E_{i}$ is the t-complete step contention-free execution fragment of
a transaction $T_{m+i}$,
$\gamma_i$ is a complete step contention-free execution fragment of
a \emph{fast-path} transaction $T_i$ such that
$\Dset(T_i)\cap \Dset(T_{m+i})=\emptyset$ in $E\cdot E_{i} \cdot \gamma_i$, and
$\rho_i$ is the execution fragment of a t-operation by $T_i$ 
that does not contain accesses to any metadata base object.
If, for all $i,j\in \{1,\ldots , m\}$, $i\neq j$,
$\Dset(T_i)\cap \Dset(T_{m+j})=\emptyset$,
$\Dset(T_i)\cap \Dset(T_{j})=\emptyset$ and
$\Dset(T_{m+i})\cap \Dset(T_{m+j})=\emptyset$, then
there exists a t-complete step 
contention-free execution fragment $E'$ that is similar to $E_{1}\cdots E_{m}$ such that 
for all $i\in \{1,\ldots, m\}$, $E\cdot E' \cdot \gamma_i\cdot  \rho_i$ is 
an execution of $\M$.
\end{lemma}
\begin{proof}
Observe that any two transactions in the execution fragment
$E_{1}\cdots E_{m}$ access mutually disjoint data sets.
Since $\M$ is progressive and provides OF TM-liveness, there exists a t-sequential execution fragment $E'=E'_{1}\cdots E'_{m}$
such that, for all $i\in \{1,\ldots , m\}$, the execution fragments $E_{i}$ and $E'_{i}$ are similar and
$E\cdot E'$ is an execution of $\M$.
Corollary~\ref{lm:pgthree} implies that, for all for all $i\in \{1,\ldots , m\}$,
$\M$ has an execution of the form $E\cdot E'_{1}\cdots E'_{i}\cdots  E'_{m} \cdot \gamma_i$.
More specifically, 
$\M$ has an execution of the form
$E\cdot \gamma_i \cdot E'_{1} \cdots E'_{i}\cdots  E'_{m}$.
Recall that the execution fragment $\rho_i$ of fast-path transaction $T_i$ that extends $\gamma_i$ contains accesses only to
base objects in $\bigcup\limits_{X\in DSet(T_i)} \mathbb{D}_X$.
Moreover, for all $i,j\in \{1,\ldots , m\}$; $i\neq j$,
$\Dset(T_i)\cap \Dset(T_{m+j})=\emptyset$ and $\Dset(T_{m+i})\cap \Dset(T_{m+j})=\emptyset$.

It follows that
$\M$ has an execution of the form
$E\cdot \gamma_i \cdot E'_{1} \cdots E'_{i}\cdot \rho_i \cdot E'_{i+1} \cdots  E'_{m}$.
and the states of each of the base objects $\bigcup\limits_{X\in DSet(T_i)} \mathbb{D}_X$ accessed by $T_i$
in the configuration after 
$E\cdot \gamma_i \cdot E'_{1} \cdots E'_{i}$
and $E\cdot \gamma_i \cdot E_{i}$ are the same.
But $E\cdot \gamma_i \cdot E_{i}\cdot \rho_i$ is an execution of $\M$.
Thus, for all $i\in \{1,\ldots , m\}$, $\M$ has an execution of the form
$E\cdot E'\cdot \gamma_i \cdot \rho_i$.
\end{proof}
Finally, we are now ready to derive our lower bound.

Let $\kappa$ be the smallest integer such that some fast-path transaction
running step contention-free after a t-complete execution performs
$\kappa$ t-reads and incurs a capacity abort. In other words, if a
fast-path transaction reads less than $\kappa$ t-objects, it cannot
incur a capacity abort.  

We prove that, for all $m \leq \kappa-1$,
there exists a t-complete execution $E_{m}$ and a set $S_m$
($|S_m|=2^{\kappa-m}$) of read-only fast-path transactions such that
(1)~each transaction in $S_m$ reads $m$ t-objects,
(2)~the data sets of any two transactions in $S_m$ are disjoint, 
(3)~in the step contention-free execution of any transaction in $S_m$
extending $E_m$, every t-read accesses at least one distinct metadata
base object. 

By induction, we assume that the induction statement holds for
all  $m <\kappa -1$ (the base case $m=0$ is trivial) and prove that
$E_{m+1}$ and $S_{m+1}$ satisfying the condition above exist.
Pick any two transactions from the set $S_m$.
We construct $E_m'$, a t-complete extension of $E_m$ by the execution
of a slow-path transaction writing to two distinct t-objects $X$ and $Y$, such that 
the two picked transactions, running step contention-free after that,
cannot distinguish $E_m$ and $E_m'$. 
Now we let each of the transactions read one of the two 
t-objects $X$ and $Y$. We show that at least one of them must access a new metadata
base object in this $(m+1)^{\textit{th}}$ t-read (otherwise, the resulting execution
would not be opaque).     
By repeating this argument for each pair of transactions, 
we derive that there exists $E_{m+1}$, a t-complete extension of
$E_m$, such that at least half of the transaction in $S_m$ must
access a new distinct metadata base object in its
$(m+1)^{\textit{th}}$ t-read when it runs t-sequentially after $E_{m+1}$. 
Intuitively, we construct $E_{m+1}$ by
``gluing'' all these executions $E_m'$ together, which is possible
thanks to Lemma~\ref{lm:pgtwo}.
These transactions constitute $S_{m+1}\subset S_m$, $|S_{m+1}|=|S_m|/2=2^{\kappa-(m+1)}$.
\begin{theorem}
\label{linearlowerbound}
Let $\M$ be any progressive, opaque HyTM implementation that provides OF TM-liveness.
For every $m\in\Nat$, there
exists an execution $E$ in which some fast-path read-only
transaction $T_k\in \ms{txns}(E)$ satisfies either (1) $\Dset(T_k) \leq m$ and $T_k$ incurs a capacity abort in $E$ or 
(2) $\Dset(T_k) = m$ and $T_k$ accesses $\Omega(m)$ distinct metadata base objects in $E$.
\end{theorem}
Here is a high-level overview of the proof technique. 
Let $\kappa$ be the smallest integer such that some fast-path transaction
running step contention-free after a t-quiescent configuration performs $\kappa$ t-reads and incurs a capacity abort. 

We prove that, for all $m \leq \kappa-1$,
there exists a t-complete execution $E_{m}$ and a set $S_m$ with $|S_m|=2^{\kappa-m}$ of
read-only fast-path transactions that access mutually disjoint data sets such that each
transaction in $S_m$
that runs step contention-free from $E_{_m}$ and 
performs t-reads of $m$ distinct t-objects accesses at least one distinct metadata base object
within the execution of each t-read operation. 

We proceed by induction. Assume that the induction statement holds for all $m <kappa -1$.
We prove that a set $S_{m+1}$; $|S_{m+1}| = 2^{\kappa-(m+1)}$ of fast-path transactions, each of which run 
step contention-free after the same t-complete execution $E_{m+1}$, perform
$m+1$ t-reads of distinct t-objects so that at least one distinct metadata base object is accessed within the execution
of each t-read operation.
In our construction, we pick any two new transactions from the set $S_m$ and show that one of them
running step contention-free from a t-complete execution that extends $E_m$ performs
$m+1$ t-reads of distinct t-objects so that at least one distinct metadata base object is accessed within the execution
of each t-read operation.
In this way, the set of transactions is reduced by half in each step of the induction 
until one transaction remains which must have accessed a distinct metadata base object in every one of its $m+1$ t-reads.

Intuitively, since all the transactions that we use in our construction access mutually disjoint data sets, 
we can apply Lemma~\ref{lm:pgtwo} to construct a t-complete execution $E_{m+1}$ such that
each of the fast-path transactions in $S_{m+1}$ when running step contention-free after $E_{m+1}$ perform
$m+1$ t-reads so that at least one distinct metadata base object is accessed within the execution of each t-read operation.

We now present the formal proof:
\begin{proof}
In the constructions which follow, every fast-path transaction executes at most $m+1$ t-reads. 
Let $\kappa$ be the smallest integer such that some fast-path transaction
running step contention-free after a t-quiescent configuration performs $\kappa$ t-reads and incurs a capacity abort. 
We proceed by induction.

\vspace{1mm}\noindent\textbf{Induction statement.}
We prove that, for all $m \leq \kappa-1$,
there exists a t-complete execution $E_{m}$ and a set $S_m$ with $|S_m|=2^{\kappa-m}$ of
read-only fast-path transactions that access mutually disjoint data sets such that each
transaction $T_{f_i}\in S_m$
that runs step contention-free from $E_{_m}$ and 
performs t-reads of $m$ distinct t-objects accesses at least one distinct metadata base object
within the execution of each t-read operation. 
Let $E_{f_i}$ be the step contention-free execution of $T_{f_i}$ after $E_{m}$ and 
let $\Dset(T_{f_i}) = \{X_{i,1}, \dots, X_{i,m}\}$. 

\vspace{1mm}\noindent\textbf{The induction.}
Assume that the induction statement holds for all $m \leq \kappa-1$.
The statement is trivially true for the base case $m=0$ for every $\kappa \in \mathbb{N}$.

We will prove that a set
$S_{m+1}$; $|S_{m+1}| = 2^{\kappa-(m+1)}$ of fast-path transactions, each of which run 
step contention-free from the same t-quiescent configuration $E_{m+1}$, perform
$m+1$ t-reads of distinct t-objects so that at least one distinct metadata base object is accessed within the execution
of each t-read operation.

The construction proceeds in \emph{phases}: there are exactly $\frac{|S_m|}{2}$ phases.
In each phase, we pick any two new transactions from the set $S_m$ and show that one of them
running step contention-free after a t-complete execution that extends $E_m$ performs
$m+1$ t-reads of distinct t-objects so that at least one distinct metadata base object is accessed within the execution
of each t-read operation.

Throughout this proof, we will assume that any two transactions (and resp. execution fragments) with distinct subscripts
represent distinct identifiers.

For all $i\in \{0,\ldots , \frac{|S_m|}{2} -1\}$, 
let $X_{2i+1},X_{2i+2}\not\in \displaystyle\bigcup_{i=0}^{|S_m|-1}\{X_{i,1},\ldots, X_{i,m}\}$ be distinct t-objects and 
let $v$ be the value of $X_{2i+1}$ and $X_{2i+2}$ after $E_m$. 
Let $T_{s_i}$ denote a slow-path transaction which writes $nv \neq v$ to $X_{2i+1}$ and $X_{2i+2}$. 
Let $E_{s_i}$ be the t-complete step contention-free execution fragment of $T_{s_i}$ running immediately after $E_m$.

Let $E'_{s_i}$ be the longest prefix of the execution $E_{s_i}$ such that 
$E_m\cdot E'_{s_i}$ can be extended neither with the 
complete step contention-free execution fragment of transaction 
$T_{f_{2i+1}}$ that performs its $m$ t-reads of $X_{2i+1,1},\ldots , X_{2i+1,m}$ and then 
performs $\Read_{f_{2i+1}}(X_{2i+1})$ and
returns $nv$, nor with the complete step contention-free execution fragment of some transaction $T_{f_{2i+2}}$ that performs 
t-reads of $X_{{2i+2}_{1}},\ldots , X_{2i+2,m}$ and then performs $\Read_{f_{2i+2}}(X_{2i+2})$ and returns $nv$.
Progressiveness and OF TM-liveness of $\M$ stipulates that such an execution exists.

Let $e_{i}$ be the enabled event of $T_{s_{i}}$ in the configuration after $E_m\cdot E'_{s_{i}}$.
By construction, the execution $E_m\cdot E'_{s_{i}}$ can be
extended with at least one of the complete step contention-free executions of transaction
$T_{f_{2i+1}}$ performing $(m+1)$ t-reads of $X_{{2i+1,1}},\ldots , X_{{2i+1,m}},X_{2i+1}$ such that
$\Read_{f_{2i+1}}(X_{2i+1})\rightarrow nv$ or
transaction $T_{f_{2i+2}}$ performing 
t-reads of $X_{{2i+2,1}},\ldots , X_{{2i+2,m}},  X_{2i+2}$ such that $\Read_{f_{2i+2}}(X_{2i+2})\rightarrow nv$.
Without loss of generality, suppose that $T_{f_{2i+1}}$ reads the value of $X_{2i+1}$ to be $nv$ after 
$E_m\cdot E'_{0_{i}} \cdot e_{i}$.

For any $i\in \{0,\ldots , \frac{|S_m|}{2} -1\}$, we will denote by $\alpha_i$ the execution fragment which we will
construct in phase $i$. 
For any $i\in \{0,\ldots , \frac{|S_m|}{2} -1\}$, we prove that
$\M$ has an execution of the form $E_m\cdot \alpha_i$ in which
$T_{f_{2i+1}}$ (or $T_{f_{2i+2}}$) running step contention-free after
a t-complete execution that extends $E_m$ performs $m+1$ t-reads of distinct t-objects so that at least one distinct metadata
base object is accessed within the execution of each first $m$ t-read operations and $T_{f_{2i+1}}$ (or $T_{f_{2i+2}}$) 
is poised to apply an event
after $E_m\cdot \alpha_i$
that accesses a distinct metadata base object during the $(m+1)^{th}$ t-read. 
Furthermore, we will show that $E_m\cdot \alpha_i$ appears t-sequential to $T_{f_{2i+1}}$ (or $T_{f_{2i+2}}$). 

\vspace{1mm}\noindent\textit{(Construction of phase $i$)}

Let $E_{f_{2i+1}}$ (and resp. $E_{f_{2i+2}}$) be the complete step contention-free execution of the t-reads of 
$X_{{2i+1},1},\dots,X_{{2i+1},m}$
(and resp. $X_{{2i+2},1},\ldots , X_{{2i+2},m}$)
running after $E_{m}$ by $T_{f_{2i+1}}$ (and resp. $T_{f_{2i+2}}$). 
By the inductive hypothesis, transaction $T_{f_{2i+1}}$ (and resp. $T_{f_{2i+2}}$)
accesses $m$ distinct metadata objects in the execution $E_m\cdot E_{f_{2i+1}}$ (and resp. $E_m\cdot E_{f_{2i+2}}$).  
Recall that transaction $T_{f_{2i+1}}$ does not conflict with transaction $T_{s_i}$. 
Thus, by Corollary~\ref{lm:pgthree},
$\M$ has an execution of the form
$E_m\cdot E'_{s_{i}} \cdot e_{i} \cdot E_{f_{2i+1}}$ 
(and resp. $E_m\cdot E'_{s_{i}} \cdot e_{i} \cdot E_{f_{2i+2}}$). 

Let $E_{rf_{2i+1}}$ be the complete step contention-free execution fragment of $\Read_{f_{2i+1}}(X_{2i+1})$ that extends
$E_{2i+1}=E_m\cdot E'_{s_{i}} \cdot e_{i} \cdot E_{f_{2i+1}} $. 
By OF TM-liveness, $\Read_{f_{2i+1}}(X_{2i+1})$ must return a matching 
response in $E_{2i+1}\cdot E_{rf_{2i+1}}$.
We now consider two cases.

\vspace{1mm}\noindent\textit{\textbf{Case \RNum{1}}: Suppose $E_{rf_{2i+1}}$ accesses at least one metadata base object 
$b$ not previously accessed by $T_{f_{2i+1}}$.}

Let $E'_{rf_{2i+1}}$ be the longest prefix of $E_{rf_{2i+1}}$ which does not apply 
any primitives to any metadata base object $b$ not previously accessed by $T_{f_{2i+1}}$. 
The execution $E_m\cdot E'_{s_{i}} \cdot e_{i} \cdot E_{f_{2i+1}} \cdot E'_{rf_{2i+1}}$ 
appears t-sequential to $T_{f_{2i+1}}$ because $E_{f_{2i+1}}$ does not contend with $T_{s_{i}}$ on any base object
and any common base object accessed in the execution fragments $E'_{rx_{2i+1}}$ and $E_{s_{i}}$
by $T_{f_{2i+1}}$ and $T_{s_{i}}$ respectively must be data objects contained in $\mathbb{D}$. 
Thus, we have that
$|\Dset(T_{f_{2i+1}})| = m+1$ and that 
$T_{f_{2i+1}}$ accesses $m$ distinct metadata base objects within each of its first $m$ t-read operations
and is poised to access a distinct metadata base object during the execution of the $(m+1)^{th}$ t-read.
In this 
case, let $\alpha_i = E_m\cdot E'_{s_{i}} \cdot e_{i} \cdot E_{f_{2i+1}} \cdot E'_{rf_{2i+1}}$.

\vspace{1.5mm}\noindent\textit{\textbf{Case \RNum{2}}: Suppose 
$E_{rf_{2i+1}}$ does not access any metadata base object not previously accessed by $T_{f_{2i+1}}$.}

In this case, we will first prove the following:
\begin{claim}
\label{cl:iterationone}
$\M$ has an execution of the form 
$E_{2i+2} = E_m\cdot E'_{s_{i}}\cdot e_{i}\cdot {\bar E}_{f_{2i+1}} \cdot E_{f_{2i+2}}$ 
where ${\bar E}_{f_{2i+1}}$ is the t-complete step contention-free execution of $T_{f_{2i+1}}$ in which 
$\Read_{f_{2i+1}}(X_{2i+1})\rightarrow nv$,
$T_{f_{2i+1}}$ invokes $\TryC_{f_{2i+1}}$ and returns a matching response.
\end{claim}
\begin{proof}
Since $E_{rf_{2i+1}}$ does not contain accesses to any distinct metadata base objects,
the execution $E_m\cdot E'_{s_{i}}\cdot e_{i} \cdot E_{f_{2i+1}}\cdot E_{rf_{2i+1}}$ appears t-sequential to $T_{f_{2i+1}}$.
By definition of the event $e_{i}$, $\Read_{f_{2i+1}}(X_{2i+1})$ must access the base object to which the event $e_i$ 
applies a nontrivial
primitive and return the response $nv$ in $E'_{s_{i}}\cdot e_{i}\cdot E_{f_{2i+1}}\cdot E_{rf_{2i+1}}$.
By OF TM-liveness, it follows that $E_m\cdot E'_{s_{i}}\cdot e_{i}\cdot {\bar E}_{f_{2i+1}}$ is an execution of $\M$.

Now recall that $E_m\cdot E'_{s_i} \cdot e_{i}\cdot E_{f_{2i+2}}$ is an execution of $\mathcal{M}$ because
transactions $T_{f_{2i+2}}$ and $T_{s_{i}}$ do not conflict in this execution and thus, cannot contend on any base object.
Finally, because $T_{f_{2i+1}}$ and $T_{f_{2i+2}}$ access disjoint data sets in 
$E_m\cdot E'_{s_{i}}\cdot e_{i}\cdot {\bar E}_{f_{2i+1}} \cdot E_{f_{2i+2}}$,
by Lemma~\ref{lm:pgtwo} again, 
we have that $E_m\cdot E'_{s_{i}} \cdot e_{i}\cdot \bar{E}_{f_{2i+1}} \cdot E_{f_{2i+2}}$ is an execution of $\mathcal{M}$.
\end{proof}
Let $E_{rf_{2i+2}}$ be the complete step contention-free execution fragment of $\Read_{f_{2i+2}}(X_{2i+2})$ after 
$E_m\cdot E'_{s_{i}}\cdot e_{i}\cdot {\bar E}_{f_{2i+1}} \cdot E_{f_{2i+2}}$. 
By the induction hypothesis and Claim~\ref{cl:iterationone}, 
transaction $T_{f_{2i+2}}$ must access $m$ distinct metadata base objects
in the execution $E_m\cdot E'_{s_{i}}\cdot e_{i}\cdot {\bar E}_{f_{2i+1}} \cdot E_{f_{2i+2}}$. 

If $E_{rf_{2i+2}}$ accesses some metadata base object, then
by the argument given in Case~\RNum{1} applied to transaction $T_{f_{2i+2}}$, we get that
$T_{f_{2i+2}}$ accesses $m$ distinct metadata base objects within each of the first $m$ t-read operations
and is poised to access a distinct metadata base object during the execution of the $(m+1)^{th}$ t-read.

Thus, suppose that $E_{rf_{2i+2}}$ does not access any metadata base object previously accessed by $T_{f_{2i+2}}$. 
We claim that this is impossible and proceed to derive a contradiction. 
In particular, $E_{rf_{2i+2}}$ does not contend with $T_{s_i}$ on any metadata base object.
Consequently, the execution $E_m\cdot  E'_{s_{i}} \cdot e_{i} \cdot {\bar E}_{f_{2i+1}} \cdot E_{f_{2i+2}}$ 
appears t-sequential to $T_{x_{2i+2}}$ since 
$E_{rx_{2i+2}}$ only contends with $T_{s_{i}}$ on base objects in $\mathbb{D}$. 
It follows that $E_{2i+2}\cdot E_{rf_{2i+2}}$ must also appear t-sequential to $T_{f_{2i+2}}$ and so 
$E_{rf_{2i+2}}$ cannot abort.
Recall that the base object, say $b$, to which $T_{s_{i}}$ applies a nontrivial primitive in the event $e_{i}$
is accessed by $T_{f_{2i+1}}$ in $E_m\cdot E'_{s_{i}}\cdot e_{i}\cdot {\bar E}_{f_{2i+1}} \cdot E_{f_{2i+2}}$; 
thus, $b \in \mathbb{D}_{X_{2i+1}}$. 
Since $X_{2i+1} \not\in \Dset(T_{f_{2i+2}})$,
$b$ cannot be accessed by $T_{f_{2i+2}}$. Thus, the execution 
$E_m\cdot E'_{s_{i}}\cdot e_{i}\cdot {\bar E}_{f_{2i+1}} \cdot E_{f_{2i+2}} \cdot E_{rf_{2i+2}}$ is indistinguishable to 
$T_{f_{2i+2}}$ from the execution ${\hat E}_i\cdot E'_{s_{i}} \cdot E_{f_{2i+2}}\cdot E_{rf_{2i+2}}$
in which $\Read_{f_{2i+2}}(X_{2i+2})$ must return the response $v$ (by construction of $E'_{s_i}$).

But we observe now that the execution 
$E_m\cdot E'_{s_{i}}\cdot e_{i}\cdot {\bar E}_{f_{2i+1}} \cdot E_{f_{2i+2}}\cdot E_{rf_{2i+2}}$ is not opaque. 
In any serialization corresponding to this execution, 
$T_{s_{i}}$ must be committed and must precede $T_{f_{2i+1}}$ because $T_{f_{2i+1}}$ read $nv$ from $X_{2i+1}$.
Also, transaction $T_{f_{2i+2}}$ must precede $T_{s_{i}}$ 
because $T_{f_{2i+2}}$ read $v$ from $X_{2i+2}$. 
However $T_{f_{2i+1}}$ must precede $T_{f_{2i+2}}$ to respect real-time ordering of transactions.
Clearly, there exists no such serialization---contradiction. 

Letting $E'_{rf_{2i+2}}$ be the longest prefix of $E_{rf_{2i+2}}$ which does not access a base object $b\in \mathbb{M}$ not previously accessed by
 $T_{f_{2i+2}}$, we can let 
 $\alpha_i=E'_{s_{i}}\cdot e_{i}\cdot {\bar E}_{f_{2i+1}} \cdot E_{f_{2i+2}}\cdot E'_{rf_{2i+2}}$ in this case.

Combining Cases~\RNum{1} and~\RNum{2}, the following claim holds.
\begin{claim}
\label{cl:final}
For each $i\in \{0,\ldots , \frac{|S_m|}{2} -1\}$, $\M$ has an execution of the form 
$E_m\cdot \alpha_i$ in which
\begin{enumerate}
\item[(1)]
some fast-path transaction 
$T_{i}\in \ms{txns}(\alpha_i)$ performs t-reads of $m+1$ distinct t-objects so that at least one distinct metadata base object
is accessed within the execution of each of the first $m$ t-reads, 
$T_i$ is poised to access a distinct metadata base object
after $E_m\cdot \alpha_i$ during the execution of the $(m+1)^{th}$ t-read and the execution appears t-sequential to $T_i$,
\item[(2)] 
the two fast-path transactions in the execution fragment $\alpha_i$ do not contend on the same base object.
\end{enumerate}
\end{claim}

\vspace{1mm}\noindent\textit{(Collecting the phases)}

We will now describe how we can construct the set $S_{m+1}$ of fast-path transactions
from these $\frac{|S_m|}{2}$ phases and force each of them to access $m+1$ distinct metadata base objects when
running step contention-free after the same t-complete execution.

For each $i\in \{0,\ldots , \frac{|S_m|}{2} -1\}$,
let $\beta_i$ be the subsequence of the execution $\alpha_i$ consisting of all the events of the fast-path transaction
that is poised to access a $(m+1)^{th}$ distinct metadata base object.
Henceforth, we denote by $T_i$ the fast-path transaction that participates in $\beta_i$.
Then, from Claim~\ref{cl:final}, it follows that, for each $i\in \{0,\ldots , \frac{|S_m|}{2} -1\}$, 
$\M$ has an execution of the form $E_m\cdot E'_{s_i}\cdot e_i \cdot \beta_i$ in which
the fast-path transaction $T_i$ performs t-reads of $m+1$ distinct t-objects 
so that at least one distinct metadata base object
is accessed within the execution of each of the first $m$ t-reads, 
$T_i$ is poised to access a distinct metadata base object
after $E_m\cdot E'_{s_i}\cdot e_i \cdot \beta_i$ during the execution of the $(m+1)^{th}$ t-read and the 
execution appears t-sequential to $T_i$. 

The following result is a corollary to the above claim that is obtained by applying the definition of ``appears t-sequential''.
Recall that $E'_{s_i}\cdot e_i$ is the t-incomplete execution of slow-path transaction $T_{s_i}$ that accesses
t-objects $X_{2i+1}$ and $X_{2i+2}$.
\begin{corollary}
\label{cr:pgone}
For all $i\in \{0,\ldots , \frac{|(S_m|}{2} -1\}$, $\M$ has an execution of the form
$E_m\cdot E_{{i}}\cdot \beta_i$ such that the configuration after $E_m\cdot E_i$ is t-quiescent,
$\ms{txns}(E_i)\subseteq \{T_{s_i}\}$
and $\Dset(T_{s_i})\subseteq \{X_{2i+1},X_{2i+2}\}$ in $E_i$.
\end{corollary}
We can represent the execution $\beta_i=\gamma_{i}\cdot \rho_{i}$ where
fast-path transaction $T_i$ performs complete t-reads of $m$ distinct t-objects in $\gamma_{i}$
and then performs an incomplete t-read of the $(m+1)^{th}$ t-object in $\rho_{i}$ in which $T_i$ only
accesses base objects in $\displaystyle\bigcup_{X\in DSet(T_i)}\{X\}$.
Recall that $T_i$ and $T_{s_i}$ do not contend on the same base object
in the execution $E_m\cdot E_i\cdot \gamma_i$.
Thus, for all $i\in \{0,\ldots , \frac{|S_m|}{2} -1\}$,
$\M$ has an execution of the form $E_m\cdot \gamma_i\cdot E_i \cdot \rho_i$.

Observe that the fast-path transaction $T_i\in \gamma_i$ does not access any t-object
that is accessed by any slow-path transaction in the execution fragment $E_0\cdots E_{\frac{|S_m|}{2} -1}$.
By Lemma~\ref{lm:finallm}, there exists a t-complete step contention-free execution fragment $E'$
that is similar to $E_0\cdots E_{\frac{|S_m|}{2} -1}$
such that
for all $i\in \{0,\ldots ,  \frac{|S_m|}{2} -1\}$, $\M$ has an execution of the form
$E_{m}\cdot E'\cdot  \gamma_i \cdot \rho_i$.
By our construction, the enabled event of each fast-path transaction $T_i\in \beta_i$ in this execution
is an access to a distinct metadata base object.

Let $S_{m+1}$ denote the set of all fast-path transactions that participate in the execution fragment
$\beta_0\cdots \beta_{ \frac{|(S_m|}{2} -1}$ and $E_{m+1}=E_m\cdot E'$.
Thus, $|S_{m+1}|$ fast-path transactions, each of which run 
step contention-free from the same t-quiescent configuration, perform
$m+1$ t-reads of distinct t-objects so that at least one distinct metadata base object is accessed within the execution
of each t-read operation. This completes the proof.
\end{proof}
\subsection{A matching upper bound}
We prove that the lower bound in Theorem~\ref{linearlowerbound} is
tight by describing an `instrumentation-optimal'' 
HyTM implementation (Algorithm~\ref{alg:inswrite}) that is opaque, progressive, provides wait-free TM-liveness,
uses \emph{invisible reads}.
\begin{definition}[Invisible reads]
We say that a HyTM implementation $\M$ uses fast-path (and resp. slow-path) \emph{invisible reads} if
for every execution $E$ of $\M$ and every fast-path (and resp. slow-path) transaction $T_k\in \ms{txns}(E)$,
$E|k$ does not contain any nontrivial events.
\end{definition}
\vspace{1mm}\noindent\textbf{Base objects.}
For every t-object $X_j$, our implementation maintains a base object $v_j\in \mathbb{D}$ that stores the value of $X_j$
and a metadata base object $r_{j}$, which is a \emph{lock bit} that stores $0$ or $1$.

\vspace{1mm}\noindent\textbf{Fast-path transactions.}
For a fast-path transaction $T_k$, the $\Read_k(X_j)$ implementation first reads $r_j$ 
to check if $X_j$ is locked by a concurrent updating transaction. 
If so, it returns $A_k$,
else it returns the value of $X_j$.
Updating fast-path transactions use uninstrumented writes:
$\Write (X_j,v)$ simply stores the cached state of $X_j$ along with its value $v$ and
if the cache has not been invalidated, updates the shared memory
during $\TryC_k$ by invoking the $\ms{commit-cache}$ primitive.

\vspace{1mm}\noindent\textbf{Slow-path read-only transactions.}
Any $\Read_k(X_j)$ invoked by a slow-path transaction first reads the value of the object from $v_j$, 
checks if $r_j$ is set
and then performs \emph{value-based validation} on its entire read set to check if any of them have been modified. 
If either of these conditions is true,
the transaction returns $A_k$. Otherwise, it returns the value of $X_j$. 
A read-only transaction simply returns $C_k$ during the tryCommit.

\vspace{1mm}\noindent\textbf{Slow-path updating transactions.}
The $\Write_k(X,v)$ implementation of a slow-path transaction stores
$v$ and the current value of $X_j$ locally, 
deferring the actual update in shared memory to tryCommit. 

During $\TryC_k$, an updating slow-path transaction $T_k$ attempts to obtain exclusive write access to its 
entire write set as follows:
for every t-object $X_j \in \Wset(T_k)$, it writes $1$ to each base
object $r_{j}$ by performing a \emph{compare-and-set} (\emph{cas})
primitive that checks if the value of $r_j$ is not $1$ 
and, if so, replaces it with $1$.   
If the \emph{cas} fails, then $T_k$ releases the locks on all objects $X_{\ell}$ 
it had previously acquired 
by writing $0$ to $r_{\ell}$ and then returns $A_k$. Intuitively, if the \emph{cas} fails, some concurrent transaction
is performing a t-write to a t-object in $\Wset(T_k)$.
If all the locks on the write set were acquired successfully,
$T_k$ checks if any t-object in $\Rset(T_k)$ is concurrently being updated by another transaction
and then performs value-based validation of the read set. If a conflict is detected from the these checks,
the transaction is aborted.
Finally, $\TryC_k$ attempts to write the values of the t-objects via \emph{cas} operations.
If any \emph{cas} on the individual base objects fails, there must be a concurrent fast-path writer, and so $T_k$ rolls back the
state of the base objects that were updated, releases locks on its write set and returns $A_k$. 
The roll backs are performed with \emph{cas} operations,
skipping any which fail to allow for concurrent fast-path writes to
locked locations. Note that if a concurrent read operation of a
fast-path transaction $T_{\ell}$ finds an ``invalid'' value in $v_j$ that was
written by such transaction $T_k$ but has not
been rolled back yet, then $T_{\ell}$ either incurs a tracking set
abort later because $T_k$ has updated $v_j$ or finds $r_j$ to be $1$. 
In both cases, the read operation of $T_{\ell}$ aborts.     

The implementation uses invisible reads (no nontrivial primitives are applied by reading transactions).
Every t-operation returns a matching response within a finite number of its steps.

\vspace{1mm}\noindent\textbf{Complexity.}
Every t-read operation performed by a fast-path transaction accesses a metadata base object
once (the lock bit corresponding to the t-object), 
which is the price to pay for detecting conflicting updating slow-path
transactions. Write operations of fast-path transactions are uninstrumented. 
Thus:
%
\begin{theorem}
\label{th:inswrite}
There exists an opaque HyTM implementation that provides uninstrumented writes, invisible reads, progressiveness
and wait-free TM-liveness such that
in its every execution $E$, every read-only fast-path transaction $T\in \ms{txns}(E)$
accesses $O(|\Rset(T)|)$ distinct metadata base objects.
\end{theorem}
%
%
%

\section{Providing partial concurrency at low cost}
\label{sec:main3}
We showed that allowing fast-path transactions to run concurrently in
HyTM results in an instrumentation cost that is
proportional to the read-set size of a fast-path transaction.   
But can we run at least \emph{some} transactions 
concurrently with constant instrumentation cost, while still keeping invisible reads?  

Algorithm~\ref{alg:inswrite2} 
implements a \emph{slow-path progressive} opaque HyTM
with invisible reads and wait-free TM-liveness. 
To fast-path transactions, it only provides \emph{sequential}
TM-progress (they are only guaranteed to commit in the absence
of concurrency), but in return the algorithm is only using a single
metadata base object $\ms{fa}$ 
that is read once by a fast-path transaction and accessed twice with a \emph{fetch-and-add}
primitive by an updating slow-path transaction.
Thus, the instrumentation cost of the algorithm is constant.   

Intuitively, $\ms{fa}$ allows fast-path transactions to detect the
existence of concurrent updating slow-path transactions.
Each time an updating slow-path updating transaction tries to commit, it increments
$\ms{fa}$ and once all writes to data base objects are completed (this
part of the algorithm is identical to Algorithm~\ref{alg:inswrite})
or the transaction is aborted,
it decrements $\ms{fa}$. Therefore, $\ms{fa}\neq 0$ means
that at least one slow-path updating transaction is incomplete.  
A fast-path transaction simply checks if $\ms{fa}\neq 0$ in the
beginning and aborts if so, 
otherwise, its code is identical to that in
Algorithm~\ref{alg:inswrite}.
Note that this way, any update of $\ms{fa}$ automatically causes a tracking set abort of any
incomplete fast-path transaction.

%
\begin{theorem}
\label{th:inswrite2}
There exists an opaque HyTM implementation that provides uninstrumented writes, invisible reads,
progressiveness for slow-path transactions, sequential TM-progress for fast-path transactions and wait-free TM-liveness
such that in every its execution $E$, every fast-path transaction
accesses at most one metadata base object.
\end{theorem}
%

\section{Related work}
\label{sec:rel}
%
The notions of \emph{opacity} and \emph{progressiveness} for STMs,
adopted in this paper for HyTMs,  were 
introduced in~\cite{GK08-opacity} and~\cite{GK09-progressiveness}, respectively.

Uninstrumented HTMs may be viewed as being inherently \emph{disjoint-access parallel}, 
a notion formalized in~\cite{israeli-disjoint, AHM09}. 
As such, some of the techniques used in Theorems~\ref{instrumentation} and~\ref{linearlowerbound} 
resemble those used in~\cite{OFTM, tm-book, AHM09, attiyaH13}.
The software component of the HyTM algorithms presented in this paper is inspired by progressive STM
implementations like \cite{DSS06,norec,KR11-TR} and is subject to the lower bounds for progressive STMs
established in \cite{GK09-progressiveness,attiyaH13,tm-book,KR11-TR}. 

Circa 2005, several papers introduced HyTM implementations~\cite{unboundedhtm1, damronhytm, kumarhytm}
that integrated HTMs with variants of \emph{DSTM}~\cite{HLM+03}.
These implementations provide nontrivial concurrency between hardware
and software transactions, by instrumenting a hardware transaction's t-operations with accesses to metadata
to detect conflicting software transactions. Thus, they impose per-access instrumentation overhead on hardware transactions,
which as we prove is inherent to such HyTM designs (Theorem~\ref{linearlowerbound}).
While these HyTM implementations satisfy progressiveness, they do not provide uninstrumented writes.
However, the HyTM implementation described in Algorithm~\ref{alg:inswrite} is provably opaque, satisfies progressiveness and
provides invisible reads. Additionally, it uses uninstrumented writes and is optimal with respect
to hardware code instrumentation.

Experiments suggest that the cost of concurrency detection is a significant bottleneck 
for many HyTM implementations~\cite{MS13}, 
which serves as a major motivation for our definition of instrumentation. 
Implementations like \emph{PhTM}~\cite{phasedtm} and \emph{HybridNOrec}~\cite{hybridnorec}
overcome the per-access instrumentation cost of \cite{damronhytm,kumarhytm} by realizing that if 
one is prepared to sacrifice progress, hardware transactions need
instrumentation only at the boundaries of transactions to detect pending software transactions.
Inspired by this observation, our HyTM implementation described in Algorithm~\ref{alg:inswrite2} overcomes the lower bound of
Theorem~\ref{linearlowerbound} by allowing hardware readers to abort due to a concurrent software writer, but
maintains progressiveness for software transactions, unlike \cite{phasedtm,hybridnorec,MS13}.

Recent work has investigated alternatives to STM fallback, such as sandboxing~\cite{ALM14,CTGM14}, and fallback to \emph{reduced} hardware transactions~\cite{MS13}. These proposals are not currently covered by our framework, although we believe that our model can be extended to incorporate such techniques.

Detailed coverage on HyTM implementations and integration with HTM proposals can be found in \cite{HLR10}.
An overview of popular HyTM designs and a comparison of the TM properties and instrumentation overhead they incur
may be found in \cite{riegel-thesis}.
%


\section{Concluding remarks}
\label{sec:disc}
We have introduced an analytical model for hybrid transactional memory
that captures the notion of cached accesses as performed by hardware transactions.
We then derived lower and upper bounds in this model to capture the inherent tradeoff between the degree of concurrency 
allowed between hardware and software transactions
and the instrumentation overhead introduced on the hardware.
In a nutshell, our results say that it is impossible to completely forgo instrumentation in a sequentially consistent HyTM, and 
that any opaque HyTM implementation providing non-trivial progress either has to pay a \emph{linear} number of metadata accesses, or will have to allow slow-path transactions to \emph{abort} fast-path operations.  

Several papers  have recently proposed the use of both direct \emph{and} cached accesses 
within the same transaction to reduce the instrumentation 
overhead~\cite{riegel-thesis,hynorecriegel, kumarhytm}, although, to
the best of our knowledge, no industrial HTM currently supports this functionality.
Another recent approach proposed \emph{reduced hardware transactions}~\cite{MS13}, 
 where part of the slow-path is executed using a short hardware transaction, which allows to
 eliminate part of the instrumentation from the hardware fast-path.
We believe that our model can be extended to incorporate
such schemes as well, and we conjecture that the lower bounds established in Theorems~\ref{instrumentation} and \ref{linearlowerbound}
would also hold in the extended model.
Future work also includes 
deriving lower bounds for HyTMs satisfying wider criteria of
consistency and progress, and exploring other complexity metrics.

\newpage
\bibliography{references}

 \appendix
 \section{Progressive opaque HyTM implementation that provides uninstrumented writes and invisible reads}
\label{app:upper}
\begin{algorithm}[!h]
\caption{Progressive opaque HyTM implementation that provides uninstrumented writes and invisible reads; code for process $p_i$
executing transaction $T_k$}
\label{alg:inswrite}
\begin{algorithmic}[1]
  	\begin{multicols}{2}
  	{
  	\footnotesize
	\Part{Shared objects}{
		\State $v_j \in \mathbb{D}$, for each t-object $X_j$ 
		\State ~~~~~allows reads, writes and cas
		\State $r_{j} \in \mathbb{M}$, for each t-object $X_j$
		\State ~~~~~allows reads, writes and cas
	}\EndPart	
	\Statex
	\Part{Local objects}{
		\State $\ms{Lset}(T_k) \subseteq \Wset(T_k)$, initially empty
		\State $\ms{Oset}(T_k) \subseteq \Wset(T_k)$, initially empty
	}\EndPart
	\Statex
	\textbf{Code for slow-path transactions}
	\Statex
	\Part{\Read$_k(X_j)$}\quad\Comment{slow-path}{
		\If{$X_j \not\in \Rset_k$}
		  \State $[\textit{ov}_j,k_j] := \Read(v_j)$ \label{line:read2}
		  \State $\Rset(T_k) := \Rset(T_k)\cup\{X_j,[\textit{ov}_j,k_j]\}$ \label{line:rset}
		  \If{$r_j\neq 0$} \label{line:abort0}
		    \Return $A_k$ \EndReturn
		  \EndIf
		  \If{$\exists X_j \in Rset(T_k)$:$(\textit{ov}_j,k_j)\neq \Read(v_j)$} \label{line:valid}
			\Return $A_k$ \EndReturn
		  \EndIf
%
		  \Return $\textit{ov}_j$ \EndReturn
		\Else
		    
		  \State $\textit{ov}_j :=\Rset(T_k).\lit{locate}(X_j)$
		  \Return $\textit{ov}_j$ \EndReturn
		\EndIf
   	 }\EndPart
	\Statex
	\Part{\Write$_k(X_j,v)$}\quad\Comment{slow-path}{
		
			\State $(\textit{ov}_j,k_j) := \Read(v_j)$
			\State $\textit{nv}_j := v$
			\State $\Wset(T_k) := \Wset(T_k)\cup\{X_j,[\textit{ov}_j,k_j]\}$
			\Return $\ok$ \EndReturn
		
   	}\EndPart
	\Statex
	
	\Part{\TryC$_k$()}\quad\Comment{slow-path}{
		\If{$\Wset(T_k)= \emptyset$}
			\Return $C_k$ \EndReturn \label{line:return}
		\EndIf
		\State locked := $\lit{acquire}(\Wset(T_k))$\label{line:acq} 
		\If{$\neg$ locked} \label{line:abort2} 
	 		\Return $A_k$ \EndReturn
	 	\EndIf
		\If{$\lit{isAbortable}()$} \label{line:abort3}
			\State $\lit{release}( \ms{Lset}(T_k))$ 
			\Return $A_k$ \EndReturn
		\EndIf
		\ForAll{$X_j \in \Wset(T_k)$}
	 		 \If{ $v_j.\lit{cas}([ov_j,k_j],[\textit{nv}_j,k])$} \label{line:write}
			      \State $\ms{Oset}(T_k):=\ms{Oset}(T_k)\cup \{X_j\}$
			 \Else
			 
			      \State $\lit{undo}(\ms{Oset}(T_k))$
			 
			 \EndIf
			 
	 	\EndFor		
		\State $\lit{release}(\Wset(T_k))$   \label{line:rel}		
   		\Return $C_k$ \EndReturn
   	 }\EndPart		
	 
 	\newpage
 	\Part{Function: $\lit{acquire}(Q$)}{
   		\ForAll{$X_j \in Q$}	
			\If{$r_j.\lit{cas}(0,1)$} \label{line:acq1}
			  \State $\ms{Lset}(T_k) := \ms{Lset}(T_k)\cup \{X_j\}$
			\Else
			  \State $\lit{release}(\ms{Lset}(T_k))$
			  \Return $\false$ \EndReturn
			\EndIf
			
		\EndFor
		
		\Return $\true$ \EndReturn 
	}\EndPart		
	 \Statex
	 \Part{Function: $\lit{release}(Q)$}{
  		\ForAll{$X_j \in Q$}	
 			\State $r_j.\Write(0)$ \label{line:rel1}
		\EndFor
		\Return $ok$ \EndReturn
	}\EndPart
	
	\Statex
	
	\Part{Function: $\lit{undo}(\ms{Oset}(T_k))$}{
		\ForAll{$X_j \in \ms{Oset}(T_k)$}
		    \State $v_j.\lit{cas}([nv_j,k],[ov_j,k_j])$
		 \EndFor
		\State $\lit{release}( \ms{Wset}(T_k))$ 
		\Return $A_k$ \EndReturn
	 }\EndPart
	
	\Statex
	 \Part{Function: $\lit{isAbortable()}$ }{
		\If{$\exists X_j \in \Rset(T_k)$: $X_j\not\in \Wset(T_k)\wedge \Read(r_{j}) \neq 0$}
			\Return $\true$ \EndReturn
		\EndIf
		\If{$\exists X_j \in Rset(T_k)$:$[\textit{ov}_j,k_j]\neq \Read(v_j)$} \label{line:valid}
			\Return $\true$ \EndReturn
		\EndIf
		\Return $\false$ \EndReturn
	}\EndPart
	\Statex
	\Statex
	\textbf{Code for fast-path transactions}
	\Statex
	\Part{$\textit{read}_k(X_j)$}\quad\Comment{fast-path}{
		\State $[\textit{ov}_j,k_j] := \Read(v_j)$ \Comment{cached read} \label{line:lin1}
		
		\If{$\Read(r_{j}) \neq 0$}  
		\label{line:hread}
			\Return $A_k$ \EndReturn
		\EndIf
		
		\Return $\textit{ov}_j$ \EndReturn
		
   	 }\EndPart
	\Statex
	\Part{$\textit{write}_k(X_j,v)$}{\quad\Comment{fast-path}
		\State $\Write(v_j,[\textit{nv}_j,k])$ \Comment{cached write} \label{line:lin2}
		\Return $\ok$ \EndReturn
		
   	}\EndPart
	\Statex
	
	\Part{$\textit{tryC}_k$()}{\quad\Comment{fast-path}
		\State $\ms{commit-cache}_i$ \label{line:lin3} \Comment{returns $C_k$ or $A_k$}
   	 }\EndPart		

	}
	\end{multicols}
  \end{algorithmic}
\end{algorithm}
Let $E$ be a 
t-sequential execution.
For every operation $\Read_k(X)$ in $E$,
we define the \emph{latest written value} of $X$ as follows:
(1) If $T_k$ contains a $\Write_k(X,v)$ preceding $\Read_k(X)$,
then the latest written value of $X$ is the value of the latest such write to $X$.
(2) Otherwise, if $E$ contains a $\Write_m(X,v)$,
$T_m$ precedes $T_k$, and $T_m$ commits in $E$,
then the latest written value of $X$ is the value
of the latest such write to $X$ in $E$.
(This write is well-defined since $E$ starts with $T_0$ writing to
all t-objects.)
We say that $\Read_k(X)$ is \emph{legal} in a t-sequential execution $E$ if it returns the
latest written value of $X$, and $E$ is \emph{legal}
if every $\Read_k(X)$ in $H$ that does not return $A_k$ is legal in $E$.

For a history $H$, a \emph{completion of $H$}, denoted ${\bar H}$,
is a history derived from $H$ as follows:
\begin{enumerate}
\item 
for every incomplete t-operation $op_k$ that is a $\Read_k \vee \Write_k$ of $T_k \in \txns(H)$ in $H$, 
insert $A_k$ somewhere after the last event of $T_k$ in $E$;
otherwise if $op_k=\TryC_k$, insert $A_k$ or $C_k$ somewhere after the last event of $T_k$ 
\item 
for every complete transaction $T_k$ in the history derived in (1) that is not t-complete, 
insert $\mathit{tryC}_k\cdot A_k$ after the 
last event of transaction $T_k$.
%
\end{enumerate}
\begin{definition}[Opacity and strict serializability]
\label{def:opaque}
A finite history $H$ is \emph{opaque} if there
is a legal t-complete t-sequential history $S$,
such that
for any two transactions $T_k,T_m \in \txns(H)$,
if $T_k$ precedes $T_m$ in real-time order, then $T_k$ precedes $T_m$ in $S$, and
$S$ is equivalent to a completion of $H$~\cite{tm-book}.

A finite history $H$ is \emph{strictly serializable} if there
is a legal t-complete t-sequential history $S$,
such that
for any two transactions $T_k,T_m \in \txns(H)$,
if $T_k \prec_H^{RT} T_m$, then $T_k$ precedes $T_m$ in $S$, and
$S$ is equivalent to $\ms{cseq}(\bar H)$, where $\bar H$ is some
completion of $H$ and $\ms{cseq}(\bar H)$ is the subsequence of $\bar H$ reduced to
committed transactions in $\bar H$.

We refer to $S$ as a \emph{serialization} of $H$.
\end{definition}
\begin{lemma}
\label{lm:opacity}
Algorithm~\ref{alg:inswrite} implements an opaque TM.
\end{lemma}
\begin{proof}
Let $E$ by any execution of Algorithm~\ref{alg:inswrite}. 
Since opacity is a safety property, it is sufficient to prove that every finite execution is opaque~\cite{icdcs-opacity}.
Let $<_E$ denote a total-order on events in $E$.

Let $H$ denote a subsequence of $E$ constructed by selecting
\emph{linearization points} of t-operations performed in $E$.
The linearization point of a t-operation $op$, denoted as $\ell_{op}$ is associated with  
a base object event or an event performed during 
the execution of $op$ using the following procedure. 

\vspace{1mm}\noindent\textbf{Completions.}
First, we obtain a completion of $E$ by removing some pending
invocations or adding responses to the remaining pending invocations
as follows:
\begin{itemize}
\item
incomplete $\Read_k$, $\Write_k$ operation performed by a slow-path transaction $T_k$ is removed from $E$;
an incomplete $\TryC_k$ is removed from $E$ if $T_k$ has not performed any write to a base object $r_j$; $X_j \in \Wset(T_k)$
in Line~\ref{line:write}, otherwise it is completed by including $C_k$ after $E$.
\item
every incomplete $\Read_k$, $\TryA_k$, $\Write_k$ and $\TryC_k$ performed by a fast-path transaction $T_k$ is removed from $E$.
\end{itemize}
\vspace{1mm}\noindent\textbf{Linearization points.}
Now a linearization $H$ of $E$ is obtained by associating linearization points to
t-operations in the obtained completion of $E$.
For all t-operations performed a slow-path transaction $T_k$, linearization points as assigned as follows:
\begin{itemize}
\item For every t-read $op_k$ that returns a non-A$_k$ value, $\ell_{op_k}$ is chosen as the event in Line~\ref{line:read2}
of Algorithm~\ref{alg:inswrite}, else, $\ell_{op_k}$ is chosen as invocation event of $op_k$
\item For every $op_k=\Write_k$ that returns, $\ell_{op_k}$ is chosen as the invocation event of $op_k$
\item For every $op_k=\TryC_k$ that returns $C_k$ such that $\Wset(T_k)
  \neq \emptyset$, $\ell_{op_k}$ is associated with the first write to a base object performed by $\lit{release}$
  when invoked in Line~\ref{line:rel}, 
  else if $op_k$ returns $A_k$, $\ell_{op_k}$ is associated with the invocation event of $op_k$
\item For every $op_k=\TryC_k$ that returns $C_k$ such that $\Wset(T_k) = \emptyset$, 
$\ell_{op_k}$ is associated with Line~\ref{line:return}
\end{itemize}
For all t-operations performed a fast-path transaction $T_k$, linearization points as assigned as follows:
\begin{itemize}
\item For every t-read $op_k$ that returns a non-A$_k$ value, $\ell_{op_k}$ is chosen as the event in Line~\ref{line:lin1}
of Algorithm~\ref{alg:inswrite}, else, $\ell_{op_k}$ is chosen as invocation event of $op_k$
\item
For every $op_k$ that is a $\TryC_k$, $\ell_{op_k}$ is the $\ms{commit-cache}_k$ primitive invoked by $T_k$
\item
For every $op_k$ that is a $\Write_k$, $\ell_{op_k}$ is the event in Line~\ref{line:lin2}.
\end{itemize}
$<_H$ denotes a total-order on t-operations in the complete sequential history $H$.

\vspace{1mm}\noindent\textbf{Serialization points.}
The serialization of a transaction $T_j$, denoted as $\delta_{T_j}$ is
associated with the linearization point of a t-operation 
performed by the transaction.

We obtain a t-complete history ${\bar H}$ from $H$ as follows. 
A serialization $S$ is obtained by associating serialization points to transactions in ${\bar H}$ as follows:
for every transaction $T_k$ in $H$ that is complete, but not t-complete, 
we insert $\textit{tryC}_k\cdot A_k$ immediately 
after the last event of $T_k$ in $H$. 
\begin{itemize}
\item If $T_k$ is an updating transaction that commits, then $\delta_{T_k}$ is $\ell_{\TryC_k}$
\item If $T_k$ is a read-only or aborted transaction,
then $\delta_{T_k}$ is assigned to the linearization point of the last t-read that returned a non-A$_k$ value in $T_k$
\end{itemize}
$<_S$ denotes a total-order on transactions in the t-sequential history $S$.
\begin{claim}
\label{cl:seq}
If $T_i \prec_{H}T_j$, then $T_i <_S T_j$
\end{claim}
\begin{proof}
This follows from the fact that for a given transaction, its
serialization point is chosen between the first and last event of the transaction
implying if $T_i \prec_{H} T_j$, then $\delta_{T_i} <_{E} \delta_{T_j}$ implies $T_i <_S T_j$.
\end{proof}
\begin{claim}
\label{cl:readfrom}
$S$ is legal.
\end{claim}
\begin{proof}
We claim that for every $\Read_j(X_m) \rightarrow v$, there exists some slow-path transaction $T_i$ (or resp. fast-path)
that performs $\Write_i(X_m,v)$ and completes the event in Line~\ref{line:write} (or resp. Line~\ref{line:lin2}) such that
$\Read_j(X_m) \not\prec_H^{RT} \Write_i(X_m,v)$.

Suppose that $T_i$ is a slow-path transaction:
since $\Read_j(X_m)$ returns the response $v$, the event in Line~\ref{line:read2}
succeeds the event in Line~\ref{line:write} performed by $\TryC_i$. 
Since $\Read_j(X_m)$ can return a non-abort response only after $T_i$ writes $0$ to $r_m$ in
Line~\ref{line:rel1}, $T_i$ must be committed in $S$.
Consequently,
$\ell_{\TryC_i} <_E \ell_{\Read_j(X_m)}$.
Since, for any updating
committing transaction $T_i$, $\delta_{T_i}=\ell_{\TryC_i}$, it follows that
$\delta_{T_{i}} <_E \delta_{T_{j}}$.

Otherwise if $T_i$ is a fast-path transaction, then clearly $T_i$ is a committed transaction in $S$.
Recall that $\Read_j(X_m)$ can read $v$ during the event in Line~\ref{line:read2}
only after $T_i$ applies the $\ms{commit-cache}$ primitive.
By the assignment of linearization points, 
$\ell_{\TryC_i} <_E \ell_{\Read_j(X_m)}$ and thus, $\delta_{T_{i}} <_E \ell_{\Read_j(X_m)}$.

Thus, to prove that $S$ is legal, it suffices to show that  
there does not exist a
transaction $T_k$ that returns $C_k$ in $S$ and performs $\Write_k(X_m,v')$; $v'\neq v$ such that $T_i <_S T_k <_S T_j$. 

$T_i$ and $T_k$ are both updating transactions that commit. Thus, 
\begin{center}
($T_i <_S T_k$) $\Longleftrightarrow$ ($\delta_{T_i} <_{E} \delta_{T_k}$) \\
($\delta_{T_i} <_{E} \delta_{T_k}$) $\Longleftrightarrow$ ($\ell_{\TryC_i} <_{E} \ell_{\TryC_k}$) 
\end{center}
Since, $T_j$ reads the value of $X$ written by $T_i$, one of the following is true:
$\ell_{\TryC_i} <_{E} \ell_{\TryC_k} <_{E} \ell_{\Read_j(X_m)}$ or
$\ell_{\TryC_i} <_{E} \ell_{\Read_j(X_m)} <_{E} \ell_{\TryC_k}$.

Suppose that $\ell_{\TryC_i} <_{E} \ell_{\TryC_k} <_{E} \ell_{\Read_j(X_m)}$.

(\textit{Case \RNum{1}:}) $T_i$ and $T_k$ are slow-path transactions.

Thus, $T_k$ returns a response from the event in Line~\ref{line:acq} 
before the read of the base object associated with $X_m$ by $T_j$ in Line~\ref{line:read2}. 
Since $T_i$ and $T_k$ are both committed in $E$, $T_k$ returns \emph{true} from the event in
Line~\ref{line:acq} only after $T_i$ writes $0$ to $r_{m}$ in Line~\ref{line:rel1}.

If $T_j$ is a slow-path transaction, 
recall that $\Read_j(X_m)$ checks if $X_j$ is locked by a concurrent transaction, 
then performs read-validation (Line~\ref{line:abort0}) before returning a matching response. 
We claim that $\Read_j(X_m)$ must return $A_j$ in any such execution.

Consider the following possible sequence of events: 
$T_k$ returns \emph{true} from \emph{acquire} function invocation, 
updates the value of $X_m$ to shared-memory (Line~\ref{line:write}), 
$T_j$ reads the base object $v_m$ associated with $X_m$, 
$T_k$ releases $X_m$ by writing $0$ to $r_{m}$ and finally $T_j$ performs the check in Line~\ref{line:abort0}. 
But in this case, $\Read_j(X_m)$ is forced to return the value $v'$ written by $T_m$--- 
contradiction to the assumption that $\Read_j(X_m)$ returns $v$. 

Otherwise suppose that $T_k$ acquires exclusive access to $X_m$ by writing $1$ to $r_{m}$ and returns \emph{true}
from the invocation of \emph{acquire}, updates $v_m$ in Line~\ref{line:write}), 
$T_j$ reads $v_m$, $T_j$ performs the check in Line~\ref{line:abort0} and finally $T_k$ 
releases $X_m$ by writing $0$ to $r_{m}$. 
Again, $\Read_j(X_m)$ must return $A_j$ since $T_j$ reads that $r_{m}$ is $1$---contradiction.

A similar argument applies to the case that $T_j$ is a fast-path transaction.
Indeed, since every \emph{data} base object read by $T_j$ is contained in its tracking set, if any concurrent
transaction updates any t-object in its read set, $T_j$ is aborted immediately by our model(cf. Section~\ref{sec:hytm}).

Thus, $\ell_{\TryC_i} <_E \ell_{\Read_j(X)} <_{E} \ell_{\TryC_k}$.

(\textit{Case \RNum{2}:}) $T_i$ is a slow-path transaction and $T_k$ is a fast-path transaction.
Thus, $T_k$ returns $C_k$ 
before the read of the base object associated with $X_m$ by $T_j$ in Line~\ref{line:read2}, but after the response
of \emph{acquire} by $T_i$ in Line~\ref{line:acq}.
Since $\Read_j(X_m)$ reads the value of $X_m$ to be $v$ and not $v'$, $T_i$ performs the \emph{cas}
to $v_m$ in Line~\ref{line:write} after the $T_k$ performs the $\ms{commit-cache}$ primitive (since if
otherwise, $T_k$ would be aborted in $E$).
But then the \emph{cas} on $v_m$ performed by $T_i$ would return $\false$ and $T_i$ would return $A_i$---contradiction.

(\textit{Case \RNum{3}:}) $T_k$ is a slow-path transaction and $T_i$ is a fast-path transaction.
This is analogous to the above case.

(\textit{Case \RNum{4}:}) $T_i$ and $T_k$ are fast-path transactions.
Thus, $T_k$ returns $C_k$ 
before the read of the base object associated with $X_m$ by $T_j$ in Line~\ref{line:read2}, but before $T_i$
returns $C_i$ (this follows from Observations~\ref{ob:one} and \ref{ob:two}).
Consequently, $\Read_j(X_m)$ must read the value of $X_m$ to be $v'$ and return $v'$---contradiction.

We now need to prove that $\delta_{T_{j}}$ indeed precedes $\ell_{\TryC_k}$ in $E$.

Consider the two possible cases:
\begin{itemize}
\item
Suppose that $T_j$ is a read-only transaction. 
Then, $\delta_{T_j}$ is assigned to the last t-read performed by $T_j$ that returns a non-A$_j$ value. 
If $\Read_j(X_m)$ is not the last t-read that returned a non-A$_j$ value, then there exists a $\Read_j(X')$ such that 
$\ell_{\Read_j(X_m)} <_{E} \ell_{\TryC_k} <_E \ell_{read_j(X')}$.
But then this t-read of $X'$ must abort by performing the checks in Line~\ref{line:abort0} or incur a tracking set abort---contradiction.
\item
Suppose that $T_j$ is an updating transaction that commits, then $\delta_{T_j}=\ell_{\TryC_j}$ which implies that
$\ell_{read_j(X)} <_{E} \ell_{\TryC_k} <_E \ell_{\TryC_j}$. Then, $T_j$ must neccesarily perform the checks
in Line~\ref{line:abort3} and return $A_j$ or incur a tracking set abort---contradiction to the assumption that $T_j$ is a committed transaction.
\end{itemize}
The proof follows.
\end{proof}
The conjunction of Claims~\ref{cl:seq} and \ref{cl:readfrom} establish that Algorithm~\ref{alg:inswrite} is opaque.
\end{proof}
\begin{theorem}[Theorem~\ref{th:inswrite}]
There exists an opaque HyTM implementation $\mathcal{M}$ that provides uninstrumented writes, invisible reads, progressiveness
and wait-free TM-liveness such that
in every execution $E$ of $\mathcal{M}$, every read-only fast-path transaction $T\in \ms{txns}(E)$
accesses $O(|\Rset(T)|)$ distinct metadata base objects.
\end{theorem}
\begin{proof}
\textit{(TM-liveness and TM-progress)}
Since none of the implementations of the t-operations in Algorithm~\ref{alg:inswrite}
contain unbounded loops or waiting statements, Algorithm~\ref{alg:inswrite} provides wait-free TM-liveness
i.e. every t-operation returns a matching response after taking a finite number of steps.

Consider the cases under which a slow-path transaction $T_k$ may be aborted in any execution.
\begin{itemize}
\item
Suppose that there exists a $\Read_k(X_j)$ performed by $T_k$ that returns $A_k$
from Line~\ref{line:abort0}.
Thus, there exists a transaction
that has written $1$ to $r_{j}$ in Line~\ref{line:acq1}, but has not yet written
$0$ to $r_{j}$ in Line~\ref{line:rel1} or
some t-object in $\Rset(T_k)$ has been updated since its t-read by $T_k$.
In both cases, there exists a concurrent transaction performing a 
t-write to some t-object in $\Rset(T_k)$, thus forcing a read-write conflict.
\item
Suppose that $\TryC_k$ performed by $T_k$ that returns $A_k$
from Line~\ref{line:abort2}.
Thus, there exists a transaction
that has written $1$ to $r_{j}$ in Line~\ref{line:acq1}, but has not yet written
$0$ to $r_{j}$ in Line~\ref{line:rel1}. Thus, $T_k$ encounters write-write conflict with another
transaction that concurrently attempts to update a t-object in $\Wset(T_k)$.
\item
Suppose that $\TryC_k$ performed by $T_k$ that returns $A_k$
from Line~\ref{line:abort3}.
Since $T_k$ returns $A_k$ from Line~\ref{line:abort3} for the same reason it
returns $A_k$ after Line~\ref{line:abort0}, the proof follows.
\end{itemize}
Consider the cases under which a fast-path transaction $T_k$ may be aborted in any execution $E$.
\begin{itemize}
\item
Suppose that a $\Read_k(X_m)$ performed by $T_k$ returns $A_k$
from Line~\ref{line:hread}.
Thus, there exists a concurrent slow-path transaction that is pending in its tryCommit and 
has written $1$ to $r_m$, but not released the lock on $X_m$ i.e. $T_k$ conflicts with another transaction in $E$.
\item
Suppose that $T_k$ returns $A_k$
while performing a cached access of some base object $b$ via a trivial (and resp. nontrivial) primitive. 
Indeed, this is possible only if some concurrent transaction writes (and resp. reads or writes) to $b$.
However, two transactions $T_k$ and $T_m$ may contend on $b$ in $E$
only if there exists $X\in\Dset(T_i)\cap\Dset(T_j)$ and $X\in\Wset(T_i)\cup\Wset(T_j)$.
from Line~\ref{line:abort2}.
The same argument applies for the case when $T_k$ returns $A_k$
while performing $\ms{commit-cache}_k$ in $E$.
\end{itemize}
\textit{(Complexity)}
The implementation uses uninstrumented writes since each $\Write_k(X_m)$ simply writes to $v_m \in \mathbb{D}_{X_{m}}$
and does not access any metadata base object.
The complexity of each $\Read_k(X_m)$ is a single access to a metadata base object $r_m$ in Line~\ref{line:hread}
that is not accessed any other transaction $T_i$ unless $X_m \in \Dset(T_i)$.
while the $\TryC_k$ just calls $\ms{cache-commit}_k$ that returns $C_k$.
Thus, each read-only transaction $T_k$ accesses $O(|\Rset(T_k)|)$ distinct metadata base objects in any execution.
\end{proof}
 \section{Opaque HyTM implementation with invisible reads that is progressive only for slow-path transactions}
\label{app:upper2}
\begin{algorithm}[!h]
\caption{Opaque HyTM implementation with progressive slow-path and sequential fast-path TM-progress; code for $T_k$ by process $p_i$}
\label{alg:inswrite2}
\begin{algorithmic}[1]
  	\begin{multicols}{2}
  	{
  	\footnotesize
	\Part{Shared objects}{
		\State $v_j \in \mathbb{D}$, for each t-object $X_j$ 
		\State ~~~~~allows reads, writes and cas
		\State $r_{j} \in \mathbb{M}$, for each t-object $X_j$
		\State ~~~~~allows reads, writes and cas
		\State $\ms{fa}$, fetch-and-add object 
	}\EndPart	
	
	\Statex	
	\textbf{Code for slow-path transactions}
	\Part{\TryC$_k$()}{\quad\Comment{slow-path}
		\If{$\Wset(T_k)= \emptyset$}
			\Return $C_k$  \EndReturn
		\EndIf
				
		\State locked := $\lit{acquire}(\Wset(T_k))$ \label{line:acq2}
		\If{$\neg$ locked} 
	 		\Return $A_k$ \EndReturn
	 	\EndIf
	 	\State $\ms{fa}.\lit{add}(1)$ \label{line:inc}
		\If{$\lit{isAbortable}()$} 
			\State $\lit{release}( \ms{Lset}(T_k))$ 
			\Return $A_k$ \EndReturn
		\EndIf
		\ForAll{$X_j \in \Wset(T_k)$}
	 		 \If{ $v_j.\lit{cas}((ov_j, k_j),(\textit{nv}_j,k))$} 
			      \State $\ms{Oset}(T_k):=\ms{Oset}(T_k)\cup \{X_j\}$
			 \Else
			      \Return $\lit{undo}(\ms{Oset}(T_k))$ \EndReturn
			 \EndIf
			 
	 	\EndFor		
		\State $\lit{release}(\Wset(T_k))$ \label{line:rel2}
   		\Return $C_k$ \EndReturn
   	 }\EndPart		
	 
 	\newpage
 	
	 \Part{Function: $\lit{release}(Q)$}{
  		\ForAll{$X_j \in Q$}	
 			\State $r_j.\Write(0)$
		\EndFor
		\State $\ms{fa}.\lit{add}(-1)$ \label{line:dec}
		\Return \ok \EndReturn
	}\EndPart
	\Statex
	\Statex

	\textbf{Code for fast-path transactions}	
	
	\Part{$\textit{read}_k(X_j)$}{\quad\Comment{fast-path}
		\If{$Rset(T_k) = \emptyset$}
			\State $l \gets \Read(\ms{fa})$ \Comment{cached read} \label{line:fread}
		
			\If{$\ms{l}\neq 0$}
			    \Return $A_k$ \EndReturn
			\EndIf
		\EndIf
		\State $(\textit{ov}_j,k_j) := \Read(v_j)$ \Comment{cached read}
		
		\Return $\textit{ov}_j$ \EndReturn
		
   	 }\EndPart
	\Statex
	\Part{$\textit{write}_k(X_j,v)$}{\quad\Comment{fast-path}
		\State $v_j.\Write(\textit{nv}_j,k)$ \Comment{cached write} 
		\Return $\ok$ \EndReturn
		
   	}\EndPart
	\Statex
	
	\Part{$\textit{tryC}_k$()}{\quad\Comment{fast-path}
		\State $\ms{commit-cache}_i$ \Comment{returns $C_k$ or $A_k$}

   	 }\EndPart		

	}
	\end{multicols}
  \end{algorithmic}
\end{algorithm}

\begin{theorem}[Theorem~\ref{th:inswrite2}]
There exists an opaque HyTM implementation $\mathcal{M}$ that provides invisible reads,
progressiveness for slow-path transactions, sequential TM-progress for fast-path transactions
and wait-free TM-liveness
such that in every execution $E$ of $\mathcal{M}$, every fast-path transaction
accesses at most one metadata base object.
\end{theorem}
\begin{proof}
The proof of opacity is almost identical to the analogous proof for Algorithm~\ref{alg:inswrite} in Lemma~\ref{lm:opacity}.

As with Algorithm~\ref{alg:inswrite}, enumerating the cases under which a slow-path transaction $T_k$
returns $A_k$ proves that Algorithm~\ref{alg:inswrite2} satisfies progressiveness for slow-path transactions.
Any fast-path transaction $T_k$; $\Rset(T_k) \neq \emptyset$ reads the metadata base object $\ms{fa}$
and adds it to the process's tracking set (Line~\ref{line:fread}).
If the value of $\ms{fa}$ is not $0$, indicating that there exists a concurrent slow-path transaction pending in its
tryCommit, $T_k$ returns $A_k$. Thus, the implementation provides sequential TM-progress for fast-path transactions.

Also, in every execution $E$ of $\mathcal{M}$, no fast-path write-only transaction accesses any metadata base object
and a fast-path reading transaction accesses the metadata base object $\ms{fa}$ exactly once, during the first t-read.
\end{proof}
\end{document}